\documentclass[review]{elsarticle}
\usepackage[latin1]{inputenc} 
\usepackage{listings, tcolorbox}

\usepackage{lineno,hyperref}
\usepackage{amssymb}
\usepackage{color}
\usepackage{amsmath}
\usepackage{tikz}
\usepackage{amsfonts}
\usepackage{mathrsfs}
\usepackage{amsthm}

\usepackage{bbding}
\usepackage{mathrsfs}
\usepackage{amsmath, amsfonts, amssymb, mathrsfs, txfonts}
\usepackage{graphicx, subfigure}
\usepackage{wasysym}
\usepackage{color}
\usepackage{bm}
\usepackage{enumerate}

\usepackage{pifont}
\usepackage{txfonts}
\usepackage{amsmath}
\usepackage{graphicx}
\usepackage{amsfonts}
\usepackage{amssymb}
\usepackage{mathrsfs,psfrag,eepic,epsfig}
\usepackage{epstopdf}

\biboptions{numbers,sort&compress}
\newtheorem{thm}{Theorem}[section]

\newtheorem{prop}[thm]{Proposition}
\theoremstyle{definition}
\newtheorem{defn}[thm]{Definition}

\modulolinenumbers[5]

\journal{Journal of \LaTeX\ Templates}

\makeatletter \@addtoreset{equation}{section}
\renewcommand{\theequation}{\arabic{section}.\arabic{equation}}

\begin{document}

\begin{frontmatter}

\title{
Riemann-Hilbert problem for the sextic nonlinear Schr\"{o}dinger equation with non-zero boundary conditions}
\tnotetext[mytitlenote]{Project supported by the Fundamental Research Fund for the Central Universities under the grant No. 2019ZDPY07.\\
\hspace*{3ex}$^{*}$Corresponding author.\\
\hspace*{3ex}\emph{E-mail addresses}: 
sftian@cumt.edu.cn,
shoufu2006@126.com (S. F. Tian)
}

\author{Xin Wu, Shou-Fu Tian$^{*}$, Jin-Jie Yang and Zhi-Qiang Li}
\address{
School of Mathematics and Institute of Mathematical Physics, China University of Mining and Technology, Xuzhou 221116, People's Republic of China
}

\begin{abstract}
We consider  a matrix Riemann-Hilbert problem for the sextic nonlinear Schr\"{o}dinger equation with a non-zero boundary conditions at infinity.
Before analyzing the spectrum problem, we introduce a Riemann surface and uniformization coordinate variable in order to avoid multi-value problems. Based on a new complex plane, the direct scattering problem perform a detailed analysis of the analytical, asymptotic and symmetry properties of the Jost functions and the scattering matrix.  Then, a generalized Riemann-Hilbert problem (RHP) is successfully  established from the results of the direct scattering transform. In the inverse scattering problem, we discuss the discrete spectrum, residue condition, trace formula and theta condition under simple poles and double poles respectively, and further solve the solution of a generalized RHP. Finally, we derive the solution  of the equation for the cases of  different poles  without reflection potential. In addition, we   analyze the localized structures and dynamic behaviors of the resulting  soliton solutions by taking some appropriate values of the parameters appeared in the solutions.
\end{abstract}


\end{frontmatter}


\section{Introduction}

The nonlinear Schr\"{o}dinger (NLS) equation
\begin{align}
iq_{t}+\frac{1}{2}q_{xx}+|q|^{2}q=0,
\end{align}
is a classic physical model, which is completely integrable. It has significant physical significance in many aspects of research such as deep water waves, fluid mechanics, plasma physics, nonlinear optics, Bose-Einstein condensation. For instance, it can describe the grouping evolution of quasi-monochromatic waves slowly changing in weakly nonlinear dispersion media. Because of the rich physical meaning of the NLS equation, it has attracted a lot of attention. Especially after Zakharov and Shabat's research about the NLS equation\cite{Shabat-1972}-\cite{Zakharov-1972}, there are more and more researches to extend it \cite{1}-\cite{12}.

Furthermore, on the basis of the previous results, it is found that high-order dispersion terms and other profound physical effects have more in-depth significance in the physical research of equations. For example, pulses of shorter duration propagate in optics along the fiber \cite{Chowdury-2015}-\cite{Chen-2016}. Therefore, the classic NLS equations can no longer meet the requirement of the research and the study of the higher-order NLS equation is more worthy of being put forward and carried out. In this work, we mainly consider the sextic nonlinear Schr\"{o}dinger (sNLS) equation \cite{Ankiewicz-2016}
\begin{align}\label{Q1}
\begin{aligned}
&iq_{t}+\frac{1}{2}(q_{xx}+2|q|^{2}q)+\delta q_{xxxxxx}+\delta\left[60q^{*}|q_{x}|^{2}+50(q^{*})^{2}q_{xx}+2q^{*}_{xxxx}\right]q^{2}\\
&+\delta q\left[12q^{*}q_{xxxx}+8q_{x}q^{*}_{xxx}+22|q_{xx}|^{2}\right]
+\delta q\left[18q_{xxx}q^{*}_{x}+70(q^{*})^{2}q^{2}_{x}\right]+20\delta(q_{x})^{2}q^{*}_{xx}\\
&+10\delta q_{x}\left[5q_{xx}q^{*}_{x}+3q^{*}q_{xxx}\right]
+20\delta q^{*}q^{2}_{xx}+10\delta q^{3}\left[(q^{*}_{x})^{2}+2q^{*}q^{2}_{xx}\right]+20\delta q|q|^{6}=0,
\end{aligned}
\end{align}
where $q(x,t)$ is the complex variable about $x$ and $t$, the asterisk denotes the complex conjugation,  and $\delta$ is a real coefficient. In addition, \eqref{Q1} will be reduced to the classic nonlinear Schr\"{o}dinger equation when $\delta=0$. Some results of the sNLS equation have also been studied in various ways such as analytic breather solutions have been obtained through the Darboux transformation in \cite{Sun-2018}, and the soliton solutions also have been obtained through inverse scattering transform in \cite{Wu-2020}.

In this work, we are dedicated to study  the inverse scattering of sNLS equation with non-zero boundary conditions at infinity, i.e.,
\begin{align}
q(x,t)\rightarrow q_{\pm}e^{i\left(q^{2}_{0}+20\delta q^{6}_{0}\right)t}, \quad x\rightarrow\pm\infty,
\end{align}
here $|q_{\pm}|=q_{0}\neq0$. Different from the case of zero boundary value, it will be more difficult to establish and solve the Riemann-Hilbert problem.

It is well known that the inverse scattering transform (IST) has developed into a powerful analysis tool for solving a large class of pure and applied mathematics, which plays an indispensable important role in the field of nonlinear science. IST was first applied to solve the Korteweg-de Vries equation   by Garderner, Greene, Kruskal and Miurra (GGKM) in 1967 \cite{Gardner-1967}, and then it was found that IST can also be applied to the NLS equation after the research of Zakharov and Shabat \cite{Shabat-1972}.  After IST was further developed, it was widely used in various equations \cite{NZBC-1}-\cite{Fan-2008}.
According to what we know, the research of the sNLS equation with NZBCs \eqref{Q1} by using the inverse scattering transform has not been reported. Therefore, the main purpose of our work is to find the soliton solution of the sNLS equation and study its physical meanings, so as to enrich the research of Schr\"{o}dinger equation.

The structure of this work is given as follows.
In section 2, we mainly study the direct scattering. Riemann surface and uniformization variables are introduced to avoid double-valued problems. Based on this result, we then discuss the analytical, symmetry and asymptotic properties of the Jost function and the scattering matrix. In section 3, we establish a generalized Riemann-Hilbert problem from the basis of direct scattering, and then we also solve the sNLS equation by solving the generalized RHP. In section 4, the inverse scattering problem with simple poles is discussed. To start with, we recover the potential by reconstructing the formula based on the discrete spectrum and residue conditions. Furthermore, trace formulate and theta condition are derived. Finally, the soliton solutions are given under reflection-less potential and its propagation behavior is also illustrated graphically when selecting appropriate parameters. In section 5, similar to simple poles, the results of the inverse scattering problem with double poles are obtained. In the end, some conclusions and discussions are presented.

\section{Direct scattering transform}

In this section, similar but different from the calculation of the zero boundary conditions, we will make some transformations and introduce new definitions under the non-zero boundary conditions.

\subsection{Riemann surface and uniformization coordinate}

In this subsection, Riemann surface will be introduced to solve the problem of multi-valued functions. The corresponding uniformization variable will also be introduced to solve the complexity.

The Lax pair of \eqref{Q1} reads
\begin{align}\label{2.1}
\begin{split}
&\psi_{x}=X\psi, X=ik\sigma_{3}+Q,\\
&\psi_{t}=T\psi, T=\left(\begin{array}{cc}
    T_{11} & T_{12} \\
    T_{21} & -T_{11}  \\
 \end{array}\right),
\end{split}
\end{align}
here $\psi=\left(\psi_{1}(x,t),\psi_{2}(x,t)\right)^{T}$,
\begin{align*}
\sigma_{3}=\left(\begin{array}{cc}
    1 & 0 \\
    0 & -1  \\
 \end{array}\right),
Q=\left(\begin{array}{cc}
    0 & iq^{*}(x,t) \\
    iq(x,t) & 0  \\
 \end{array}\right),
\end{align*}

\begin{gather*}
T_{11}=i\left[-\frac{1}{2}|q|^{2}-10\delta|q|^{6}-5\delta\left(q^{2}_{x}(q^{*})^{2}\right)
-\delta|q_{xx}|^{2}
+\delta\left(q_{x}q^{*}_{xxx}-q^{*}q_{xxxx}-qq^{*}_{xxxx}\right)\right.\\
\left.-10\delta|q|^{2}\left(q_{xx}q^{*}+q^{*}_{xx}q\right)\right]
+ik\left[12i\delta|q|^{2}\left(q_{x}q^{*}-q^{*}_{x}q\right)
+2i\delta\left(q_{x}q^{*}_{xx}-q^{*}_{x}q_{xx}+q^{*}q_{xxx}-q^{*}_{xxx}q\right)\right]\\
+ik^{2}\left[1+12\delta|q|^{4}-4\delta|q_{x}|^{2}
+4\delta\left(q^{*}_{xx}q+q_{xx}q^{*}\right)\right]
-8\delta k^{3}\left(qq^{*}_{x}-q^{*}q_{x}\right)
-16i\delta k^{4}|q|^{2}+32i\delta k^{6},\\
T_{21}=32i\delta k^{5}q-16\delta k^{4}q_{x}
+ik^{3}\left(-16\delta|q|^{2}q-8\delta q_{xx}\right)
+ik^{2}\left(-24i\delta|q|^{2}q_{x}-4i\delta q_{xxx}\right)\\
ik\left[q+12\delta q^{*}q^{2}_{x}+16\delta|q|^{2}q_{xx}+4\delta q^{2}q^{*}_{xx}
+2\delta q_{xxxx}+12\delta|q|^{4}q+8\delta q|q_{x}|^{2}\right]\\
+i\left[\frac{i}{2}q_{x}+i\delta q_{xxxxx}+10i\delta\left(qq^{*}_{x}q_{xx}+qq^{*}_{xx}q_{x}
+|q|^{2}q_{xxx}+3|q|^{4}q_{x}+q_{x}|q_{x}|^{2}2q^{*}q_{x}q_{xx}\right)\right],\\
T_{12}=32i\delta k^{5}q^{*}-16\delta k^{4}q^{*}_{x}
+ik^{3}\left(-16\delta|q|^{2}q^{*}-8\delta q^{*}_{xx}\right)
+ik^{2}\left(24i\delta|q|^{2}q^{*}_{x}+4i\delta q^{*}_{xxx}\right)\\
ik\left[q^{*}+12\delta qq^{2}_{x}+16\delta|q|^{2}q^{*}_{xx}+4\delta q^{2}q_{xx}
+2\delta q^{*}_{xxxx}+12\delta|q|^{4}q^{*}+8\delta q^{*}|q_{x}|^{2}\right]\\
-i\left[\frac{i}{2}q^{*}_{x}+i\delta q^{*}_{xxxxx}+10i\delta\left(q^{*}q_{x}q^{*}_{xx}+q^{*}q_{xx}q^{*}_{x}
+|q|^{2}q^{*}_{xxx}+3|q|^{4}q^{*}_{x}+q^{*}_{x}|q_{x}|^{2}
+2qq^{*}_{x}q^{*}_{xx}\right)\right],\\
\end{gather*}
where  $k$ is the  spectral parameter. Furthermore, \eqref{Q1} satisfies zero curvature equation $U_{t}-V_{x}+[U,V]=0$ which is the compatibility condition of the Lax pair \eqref{2.1}.

In order to treat a non-zero boundary related to time $t$ into a non-zero boundary independent of time $t$, we make the following transformation
\begin{align}\label{Q2}
\left\{\begin{aligned}
&q\rightarrow qe^{i\left(q^{2}_{0}+20\delta q^{6}_{0}\right)t},\\
&\psi\rightarrow\psi e^{-i\left(\frac{1}{2}q^{2}_{0}+10\delta q^{6}_{0}\right)t\sigma_{3}}.
\end{aligned}\right.
\end{align}
Then \eqref{Q1} develops into
\begin{align}\label{Q3}
\begin{aligned}
&iq_{t}+\frac{1}{2}q_{xx}+q(|q|^{2}-q^{2}_{0})+20\delta q(|q|^{6}-q^{6}_{0})
+\delta q_{xxxxxx}+10\delta q^{3}\left[(q^{*}_{x})^{2}+2q^{*}q^{2}_{xx}\right]\\
&+\delta q\left[12q^{*}q_{xxxx}+8q_{x}q^{*}_{xxx}+22|q_{xx}|^{2}\right]
+\delta q\left[18q_{xxx}q^{*}_{x}+70(q^{*})^{2}q^{2}_{x}\right]+20\delta(q_{x})^{2}q^{*}_{xx}\\
&+10\delta q_{x}\left[5q_{xx}q^{*}_{x}+3q^{*}q_{xxx}\right]
+20\delta q^{*}q^{2}_{xx}+\delta\left[60q^{*}|q_{x}|^{2}+50(q^{*})^{2}q_{xx}+2q^{*}_{xxxx}\right]q^{2}=0,
\end{aligned}
\end{align}
and the Lax pair also becomes
\begin{align}\label{Q5}
\left\{\begin{aligned}
&\psi_{x}=X\psi, \\
&\psi_{t}=T\psi,
\end{aligned}\right.
\end{align}
here
\begin{align*}
&X=ik\sigma_{3}+Q,\\
&T=\left(\begin{array}{cc}
    T_{11}-i\left(\frac{1}{2}q^{2}_{0}+10\delta q^{6}_{0}\right) & T_{12} \\
    T_{21} & -T_{11}+i\left(\frac{1}{2}q^{2}_{0}+10\delta q^{6}_{0}\right)  \\
 \end{array}\right).
\end{align*}

The new boundary conditions is
\begin{align}\label{Q4}
\lim_{x\rightarrow\pm\infty}q(x,t)=q_{\pm}, \quad|q_{\pm}|=q_{0}\neq0,
\end{align}
so when $x\rightarrow\pm\infty$, we obtain the limit spectral problem
\begin{align}\label{Q6}
\left\{\begin{aligned}
&\psi_{x}=X_{\pm}\psi,\\
&\psi_{t}=T_{\pm}\psi,
\end{aligned}\right.
\end{align}
where
\begin{align*}
X_{\pm}=ik\sigma_{3}+Q_{\pm}, \quad Q_{\pm}=\left(\begin{array}{cc}
  0 & iq^{*}_{\pm} \\
  iq_{\pm} & 0 \\
 \end{array}\right),
\end{align*}
\begin{align*}
T_{\pm}=&32i\delta k^{6}\sigma_{3}+32\delta k^{5}Q_{\pm}+16i\delta k^{4}Q^{2}_{\pm}\sigma_{3}
+16\delta k^{3}Q^{3}_{\pm}+ik^{2}\sigma_{3}+12i\delta k^{2}Q^{4}_{\pm}\sigma_{3}\\
&+kQ_{\pm}+12i\delta kQ^{5}_{\pm}+\frac{1}{2}iQ^{2}_{\pm}\sigma_{3}+10i\delta Q^{6}_{\pm}\sigma_{3}-i\left(\frac{1}{2}q^{2}_{0}+10\delta q^{6}_{0}\right)\sigma_{3}\\
=&\left(32\delta k^{5}+16\delta k^{3}q^{2}_{0}+12\delta kq^{4}_{0}+k\right)X_{\pm}.
\end{align*}

From the calculation, we can know that $X_{\pm}$ has two eigenvalues $i\lambda$ and $-i\lambda$, where
\begin{align*}
\lambda^{2}=k^{2}+q^{2}_{0}=(k+iq_{0})(k-iq_{0}),
\end{align*}
the Riemann surface determined by this is formed by bonding complex $k$-planes $S_{1}$ and $S_{2}$ cut along secant $[-iq_{0}, iq_{0}]$, here the branch points are $k=\pm iq_{0}$. On $S_{1}$, introducing local polar coordinates
\begin{align}
k+iq_{0}=r_{1}e^{i\theta_{1}}, k-iq_{0}=r_{2}e^{i\theta_{2}},
\quad-\frac{\pi}{2}<\theta_{1}, \theta_{2}<\frac{3\pi}{2},
\end{align}
then the single-valued analytical branch function can be written on the Riemann surface as
\begin{align}
\lambda(k)=
\left\{\begin{aligned}
&(r_{1}r_{2})^{\frac{1}{2}}e^{\frac{\theta_{1}+\theta_{2}}{2}}, \qquad on ~ S_{1},\\
&-(r_{1}r_{2})^{\frac{1}{2}}e^{\frac{\theta_{1}+\theta_{2}}{2}}, \quad on ~ S_{2}.\\
\end{aligned}\right.
\end{align}
Meanwhile, it has the following properties

$\bullet$ $Imk>0$ of $S_{1}$ and $Imk<0$ of $S_{2}$ are mapped into $Im\lambda>0$;

$\bullet$ $Imk<0$ of $S_{1}$ and $Imk>0$ of $S_{2}$ are mapped into $Im\lambda<0$;

$\bullet$ The branch $[-iq_{0}, iq_{0}]$ of $S_{1}$ and $S_{2}$ is mapped into the branch $[-q_{0}, q_{0}]$ of $\lambda$-plane.

Defining the uniformization variable
\begin{align}
z=k+\lambda,
\end{align}
so we get the following relations
\begin{align}\label{Q11}
k(z)=\frac{1}{2}\left(z-\frac{q^{2}_{0}}{z}\right), \quad
\lambda(z)=\frac{1}{2}\left(z+\frac{q^{2}_{0}}{z}\right).
\end{align}
Considering the second equation of \eqref{Q11}, namely, Joukowsky transformation,
\begin{align*}
\lambda(z)=\frac{1}{2|z|^{2}}\left[(|z|^{2}-q^{2}_{0})z+2q^{2}_{0}Rez\right],
\end{align*}
one obtains
\begin{align*}
Im\lambda(z)=\frac{1}{2|z|^{2}}(|z|^{2}-q^{2}_{0})Imz.
\end{align*}
Therefore, Joukowsky transformation has the following properties:

$\bullet$ $Im\lambda>0$ is mapped into
\begin{align*}
D^{+}=\left\{z\in\mathbb{C}:(|z|^{2}-q^{2}_{0})Imz>0\right\};
\end{align*}

$\bullet$ $Im\lambda<0$ is mapped into
\begin{align*}
D^{-}=\left\{z\in\mathbb{C}:(|z|^{2}-q^{2}_{0})Imz<0\right\};
\end{align*}

$\bullet$ The branch $[-q_{0}, q_{0}]$ is mapped into the circle of $z$-plane
\begin{align*}
C_{0}=\left\{|z|=q_{0}, z\in\mathbb{C}\right\}.
\end{align*}

The transformation relationship from $k$-plane to $z$-plane is shown in the figure below.

\centerline{\begin{tikzpicture}
\path [fill=red] (-4.5,2.5) -- (-0.5,2.5) to
(-0.5,4.5) -- (-4.5,4.5);
\draw[-][thick](-4.5,2.5)--(-2.5,2.5);
\draw[fill] (-2.5,2.5) circle [radius=0.035];
\draw[->][thick](-2.5,2.5)--(-0.5,2.5)node[above]{$Rek$};
\draw[<-][thick](-2.5,4.5)node[right]{$Imk$}--(-2.5,3.5)node[right]{$iq_{0}$};
\draw[fill] (-2.5,3.5) circle [radius=0.035];
\draw[-][thick](-2.5,3.5)--(-2.5,2.5);
\draw[-][thick](-2.5,2.5)--(-2.5,1.5)node[right]{$-iq_{0}$};
\draw[fill] (-2.5,1.5) circle [radius=0.035];
\draw[-][thick](-2.5,1.5)--(-2.5,0.5);
\draw[fill] (-4.5,4) node[right]{$S_{1}$};
\draw[fill] (-2.5,2.2) node[right]{$0$};
\draw[fill] (-1.8,3.8) circle [radius=0.035] node[right]{$Im k>0$};
\draw[fill] (-1.8,1.7) circle [radius=0.035] node[right]{$Im k<0$};
\path [fill=red] (0.5,0.5) -- (4.5,0.5) to
(4.5,2.5) -- (0.5,2.5);
\draw[-][thick](0.5,2.5)--(2.5,2.5);
\draw[fill] (2.5,2.5) circle [radius=0.035];
\draw[->][thick](2.5,2.5)--(4.5,2.5)node[above]{$Rek$};
\draw[<-][thick](2.5,4.5)node[right]{$Imk$}--(2.5,3.5)node[right]{$iq_{0}$};
\draw[fill] (2.5,3.5) circle [radius=0.035];
\draw[-][thick](2.5,3.5)--(2.5,2.5);
\draw[-][thick](2.5,2.5)--(2.5,1.5)node[right]{$-iq_{0}$};
\draw[fill] (2.5,1.5) circle [radius=0.035];
\draw[-][thick](2.5,1.5)--(2.5,0.5);
\draw[fill] (0.5,4) node[right]{$S_{2}$};
\draw[fill] (-2.5,2.2) node[right]{$0$};
\draw[fill] (3.1,3.8) circle [radius=0.035] node[right]{$Im k>0$};
\draw[fill] (3.1,1.7) circle [radius=0.035] node[right]{$Im k<0$};
\path [fill=red] (-4.5,-2.5) -- (-0.5,-2.5) to
(-0.5,-0.5) -- (-4.5,-0.5);
\draw[-][thick](-4.5,-2.5)--(-2.5,-2.5);
\draw[fill] (-2.5,-2.5) circle [radius=0.035];
\draw[->][thick](-2.5,-2.5)--(-0.5,-2.5)node[above]{$Rek$};
\draw[<-][thick](-2.5,-0.5)node[right]{$Imk$}--(-2.5,-1.5)node[right]{$iq_{0}$};
\draw[fill] (-2.5,-1.5) circle [radius=0.035];
\draw[-][thick](-2.5,-1.5)--(-2.5,-2.5);
\draw[-][thick](-2.5,-2.5)--(-2.5,-3.5)node[right]{$-iq_{0}$};
\draw[fill] (-2.5,-3.5) circle [radius=0.035];
\draw[-][thick](-2.5,-3.5)--(-2.5,-4.5);
\draw[fill] (-4.5,-1) node[right]{$\lambda-$plane};
\draw[fill] (-2.5,-2.7) node[right]{$0$};
\draw[fill] (-1.9,-1.7) circle [radius=0.035] node[right]{$Im \lambda>0$};
\draw[fill] (-1.9,-3.8) circle [radius=0.035] node[right]{$Im \lambda<0$};
\path [fill=red] (0.5,-2.5) -- (4.5,-2.5) to
(4.5,-0.5) -- (0.5,-0.5);
\filldraw[white, line width=0.5](3.5,-2.5) arc (0:180:1);
\filldraw[red, line width=0.5](1.5,-2.5) arc (-180:0:1);
\draw[->][thick](0.5,-2.5)--(1,-2.5);
\draw[-][thick](1,-2.5)--(2,-2.5);
\draw[<-][thick](2,-2.5)--(2.5,-2.5);
\draw[fill] (2.5,-2.5) circle [radius=0.035];
\draw[-][thick](2.5,-2.5)--(3,-2.5);
\draw[<->][thick](3,-2.5)--(4,-2.5);
\draw[-][thick](4,-2.5)--(4.5,-2.5)node[above]{$Rez$};
\draw[-][thick](2.5,-0.5)node[right]{$Imz$}--(2.5,-2.5);
\draw[-][thick](2.5,-2.5)--(2.5,-4.5);
\draw[fill] (2.5,-2.8) node[right]{$0$};
\draw[fill] (2.5,-1.5) circle [radius=0.035];
\draw[fill] (2.5,-3.5) circle [radius=0.035];
\draw[fill] (2.5,-1.2) node[right]{$iq_{0}$};
\draw[fill] (2.5,-3.8) node[right]{$-iq_{0}$};
\draw[-][thick](3.5,-2.5) arc(0:360:1);
\draw[-<][thick](3.5,-2.5) arc(0:30:1);
\draw[-<][thick](3.5,-2.5) arc(0:150:1);
\draw[->][thick](3.5,-2.5) arc(0:210:1);
\draw[->][thick](3.5,-2.5) arc(0:330:1);
\draw[->][thick](-1.5,1.5)--(-1.5,-0.3);
\draw[->][thick](1.5,1.5)--(-1,-0.4);
\draw[->][thick](-0.8,- 3.2)--(0.8,-3.2);
\end{tikzpicture}}
\noindent { \small \textbf{Figure 1.}
Transformation relationship from $k$-Riemann plane to $\lambda$-plane and $z$-plane.
}\\

In addition, for $k\rightarrow\infty$, $z$ has two asymptotic states:
\begin{align*}
&on ~S_{1}, k\rightarrow\infty \Rightarrow z\rightarrow\infty;\\
&on ~S_{2}, k\rightarrow\infty \Rightarrow z\rightarrow0.
\end{align*}

\subsection{Jost functions}
In this subsection, some results of Jost functions are obtained.

Through the previous related calculations, we know that $X_{\pm}$ has two eigenvalues $\pm i\lambda$ and $T_{\pm}$ also has two eigenvalues $\pm i\lambda(32\delta k^{5}+16\delta k^{3}q^{3}_{0}+12\delta kq^{4}_{0}+k)$. According to the relation of $X_{\pm}$ and $T_{\pm}$, it is obviously that they can be diagonalized by the same matrix, i.e.,
\begin{align}\label{Q14}
\left\{\begin{aligned}
&X_{\pm}(x,t;z)=Y_{\pm}(z)(i\lambda\sigma_{3})Y^{-1}_{\pm}(z),\\
&T_{\pm}(x,t;z)=Y_{\pm}(z)(i\lambda(32\delta k^{5}+16\delta k^{3}q^{3}_{0}+12\delta kq^{4}_{0}+k)\sigma_{3})Y^{-1}_{\pm}(z),
\end{aligned}\right.
\end{align}
where
\begin{align*}
Y_{\pm}(z)=\left(
             \begin{array}{cc}
               1 & -\frac{q^{*}_{\pm}}{z} \\
               \frac{q_{\pm}}{z} & 1 \\
             \end{array}
           \right)
=\mathbb{I}+\frac{i}{z}\sigma_{3}Q_{\pm}.
\end{align*}

Substituting \eqref{Q14} into \eqref{Q6} gets
\begin{align}
\left\{\begin{aligned}
&(Y^{-1}_{\pm}\psi)_{x}=i\lambda\sigma_{3}Y^{-1}_{\pm}\psi,\\
&(Y^{-1}_{\pm}\psi)_{t}=i\lambda(32\delta k^{5}+16\delta k^{3}q^{3}_{0}+12\delta kq^{4}_{0}+k)\sigma_{3}Y^{-1}_{\pm}\psi.
\end{aligned}\right.
\end{align}
Then, the solution of asymptotic spectral problem \eqref{Q6} is
\begin{align}
\psi=Y_{\pm}e^{i\theta(z)\sigma_{3}},
\end{align}
where $\theta(z)=\lambda(z)\left[x+(32\delta k^{5}+16\delta k^{3}q^{3}_{0}+12\delta kq^{4}_{0}+k)t\right]$.

Therefore, the Lax pair \eqref{Q5} has a solution asymptotically to $\psi$
\begin{align}\label{Q17}
\psi_{\pm}\sim Y_{\pm}e^{i\theta(z)\sigma_{3}}, \quad x\rightarrow\pm\infty.
\end{align}
For \eqref{Q5}, making the following transformation
\begin{align}\label{Q18}
\mu_{\pm}=\psi_{\pm}e^{-i\theta(z)\sigma_{3}},
\end{align}
 we obtain
\begin{align}\label{Q19}
\mu_{\pm}\sim Y_{\pm}, \quad x\rightarrow\pm\infty.
\end{align}
The equivalent Lax pair of \eqref{Q5} is
\begin{align}\label{Q20}
\left\{\begin{aligned}
&(Y^{-1}_{\pm}\mu_{\pm})_{x}+i\lambda\left[Y^{-1}_{\pm}\mu_{\pm}, \sigma_{3}\right]
=Y^{-1}_{\pm}\triangle Q_{\pm}\mu_{\pm},\\
&(Y^{-1}_{\pm}\mu_{\pm})_{t}+i\lambda(32\delta k^{5}+16\delta k^{3}q^{3}_{0}+12\delta kq^{4}_{0}+k)\left[Y^{-1}_{\pm}\mu_{\pm}, \sigma_{3}\right]
=Y^{-1}_{\pm}\triangle T_{\pm}\mu_{\pm},
\end{aligned}\right.
\end{align}
with $\triangle Q_{\pm}=Q-Q_{\pm}$, $\triangle T_{\pm}=T-T_{\pm}$.

Furthermore, the full derivative form of the equivalent Lax pair \eqref{Q20} reads
\begin{align}
d\left(e^{i\theta(z)\hat{\sigma}_{3}}Y^{-1}_{\pm}\mu_{\pm}\right)
=e^{i\theta(z)\hat{\sigma}_{3}}\left[Y^{-1}_{\pm}\left(\triangle Q_{\pm}dx-\triangle T_{\pm}dt\right)\mu_{\pm}\right],
\end{align}
where $e^{\hat{\sigma}_{3}}A=e^{\sigma_{3}}Ae^{-\sigma_{3}}$.
By integrating along two special paths $(-\infty,t)\rightarrow(x,t)$ and $(+\infty,t)\rightarrow(x,t)$, the following two Volterra integral equations can be derived as
\begin{align}\label{Q22}
\left\{\begin{aligned}
&\mu_{-}(x,t;z)=Y_{-}+\int^{x}_{-\infty}Y_{-}e^{i\lambda(x-y)\hat{\sigma}_{3}}
\left[Y^{-1}_{-}\triangle Q_{-}(y,t;z)\mu_{-}(y,t;z)\right]dy,\\
&\mu_{+}(x,t;z)=Y_{+}-\int^{+\infty}_{x}Y_{+}e^{i\lambda(x-y)\hat{\sigma}_{3}}
\left[Y^{-1}_{+}\triangle Q_{+}(y,t;z)\mu_{+}(y,t;z)\right]dy.
\end{aligned}\right.
\end{align}
If $q-q_{\pm}\in L^{1}(\mathbb{R})$, then from the first of the Volterra integral equations, we get
\begin{align}\label{Q23}
Y^{-1}_{-}\mu_{-,1}=\left(\begin{array}{c}
                                   1 \\
                                   0
                                 \end{array}\right)
+\int^{x}_{-\infty}G(x-y,z)\triangle Q_{-}(y)\mu_{-,1}dy,
\end{align}
here $x-y>0$ and
\begin{align*}
G(x-y,z)=\frac{1}{\gamma}\left(
                           \begin{array}{cc}
                             1 & \frac{q^{*}_{-}}{z} \\
                             -\frac{q_{-}}{z}e^{-2i\lambda(x-y)} & e^{-2i\lambda(x-y)} \\
                           \end{array}
                         \right),
\end{align*}
with $\gamma=det(Y_{\pm})=1+\frac{q^{2}_{0}}{z^{2}}$.
Noticing
\begin{align*}
e^{-2i\lambda(x-y)}=e^{-2i(x-y)Re\lambda}e^{2(x-y)Im\lambda},
\end{align*}
it is obviously that $\mu_{-,1}$ is analytic in $D^{-}$. In the same way, we obtain $\mu_{+,2}$ is also analytic in $D^{-}$, while $\mu_{+,1}$ and $\mu_{-,2}$ are analytic in $D^{+}$. Furthermore, they can be recorded as
\begin{align*}
&\mu_{+}=\left(\mu^{+}_{+,1}, \mu^{-}_{+,2}\right),\\
&\mu_{-}=\left(\mu^{-}_{-,1}, \mu^{+}_{-,2}\right).
\end{align*}

\subsection{Scattering matrix}

In this section, some results of scattering matrix are obtained.

Before the beginning of this section, we first introduce Abel's theorem.
\begin{thm}
Assuming $A(x)\in\mathbb{C}^{n\times n}$, for
\begin{align*}
Y_{x}=A(x)Y,
\end{align*}
then
\begin{align*}
(detY)_{x}=trAdetY,
\end{align*}
therefore
\begin{align*}
detY(x)=detY(x_{0})e^{\int^{x}_{x_{0}}trA(t)dt}.
\end{align*}
\end{thm}

Obviously, from the expressions of $X$ and $T$ in the Lax pair, we know
\begin{align}
trX=trT=0.
\end{align}
Again according to Abel's theorem, it can be deduced that
\begin{align}
\left(det \psi_{\pm}\right)_{x}=\left(det \psi_{\pm}\right)_{t}=0.
\end{align}
Thus, through \eqref{Q17} we know
\begin{align}
det\psi_{\pm}=detY_{\pm}=\gamma.
\end{align}

Due to $\psi_{\pm}$ are the solution of Lax pair \eqref{Q5}, they have a linear relationship each other, there exists $S(z)$ satisfying
\begin{align}\label{Q28}
\psi_{+}(x,t;z)=\psi_{-}(x,t;z)S(z),
\end{align}
where $S(z)=(s_{ij})_{2\times 2}$ does not depend on $x$ and $t$.

Furthermore, it can be written in component form as
\begin{align}
&\psi_{+,1}=s_{11}\psi_{-,1}+s_{21}\psi_{-,2},\label{Q29a}\\
&\psi_{+,2}=s_{12}\psi_{-,1}+s_{22}\psi_{-,2}.\label{Q29b}
\end{align}

Defining Wronskian as $W(\phi,\psi)=\phi_{1}\psi_{2}-\phi_{2}\psi_{1}$, it is obviously that Wronskian satisfies the following properties
\begin{align*}
&W(\phi,\psi)=-W(\psi,\phi),\\
&W(c_{1}\phi,c_{2}\psi)=c_{1}c_{2}W(\phi,\psi).
\end{align*}
By the component form of \eqref{Q28}, the follow results can be obtained by
\begin{align}\label{Q30}
\begin{split}
&s_{11}=W(\psi_{+,1},\psi_{-,2})/\gamma, \quad s_{12}=W(\psi_{+,2},\psi_{-,2})/\gamma,\\
&s_{21}=W(\psi_{-,1},\psi_{+,1})/\gamma, \quad s_{22}=W(\psi_{-,1},\psi_{+,2})/\gamma.
\end{split}
\end{align}
Combining \eqref{Q18} and \eqref{Q28} deduces
\begin{align}\label{Q31}
\mu_{+}=\mu_{-}e^{i\theta\hat{\sigma}_{3}}S(z).
\end{align}
According to $det(\mu_{\pm})_{x}=det(\mu_{\pm})_{t}=0$, then
\begin{align}
det(\mu_{\pm})=\lim_{x\rightarrow\pm\infty}det(\mu_{\pm})=det(Y_{\pm})=\gamma\neq0.
\end{align}
Thus, $\mu_{\pm}$ is reversible. Direct calculation shows that $S(z)$ can be represented by $\mu_{\pm}$. Assuming $q-q_{\pm}\in L^{1}(\mathbb{R})$, according to the analyticity of $\mu_{\pm}$, we know that $s_{11}$ is analytic on $D^{+}$, $s_{22}$ is analytic on $D^{-}$, and $s_{12}$ and $s_{21}$ are continue to $\Sigma$.

In addition, we define the reflection coefficients as
\begin{align}
\rho(z)=\frac{s_{21}(z)}{s_{11}(z)}, \quad \tilde{\rho}(z)=\frac{s_{12}(z)}{s_{22}(z)}.
\end{align}

\subsection{Symmetris}

According to the first of the equivalent lax pair \eqref{Q20}, one obtains
\begin{align}\label{Q34}
\left[Y^{-1}_{\pm}(z)\mu_{\pm}(z)\right]_{x}+i\lambda(z)
\left[Y^{-1}_{\pm}(z)\mu_{\pm}(z)\sigma_{3}-\sigma_{3}Y^{-1}_{\pm}(z)\mu_{\pm}(z)\right]
=Y^{-1}_{\pm}(z)\triangle Q_{\pm}(z)\mu_{\pm}(z).
\end{align}
For the above equation, we replace $z$ with $z^{*}$ and take the conjugate, and multiply $\sigma$ on both sides. Then By simple calculation one can obtain
\begin{align}
\begin{aligned}
&\left[Y^{-1}_{\pm}(z)\sigma\mu^{*}_{\pm}(z^{*})\sigma\right]_{x}+i\lambda(z)
\left[Y^{-1}_{\pm}(z)\sigma\mu^{*}_{\pm}(z^{*})\sigma\sigma_{3}
-\sigma_{3}Y^{-1}_{\pm}(z)\sigma\mu^{*}_{\pm}(z^{*})\sigma\right]\\
&=Y^{-1}_{\pm}(z)\triangle Q_{\pm}(z)\sigma\mu^{*}_{\pm}(z^{*})\sigma,
\end{aligned}
\end{align}
here $\sigma=\left(\begin{array}{cc}
               0 & 1 \\
               -1 & 0
             \end{array}\right)
$.
Due to \eqref{Q19}, it is easy to see that
\begin{align}
-\sigma\mu^{*}_{\pm}(z^{*})\sigma\sim Y_{\pm}, \quad x\rightarrow\pm\infty.
\end{align}
Thus, $\mu_{\pm}(z)$ and $-\mu^{*}_{\pm}(z^{*})$ satisfy the same equation and have the same asymptotic behavior, i.e., they are equal. From this, we get the symmetry relation
\begin{align}
\mu_{\pm}(z)=-\sigma\mu^{*}_{\pm}(z^{*})\sigma.
\end{align}
Similarly, the following symmetry relation can be obtained as
\begin{align}
\mu_{\pm}(z)=\frac{i}{z}\mu_{\pm}\left(-\frac{q^{2}_{0}}{z}\right)\sigma_{3}Q_{\pm}.
\end{align}
Expanding the symmetry relationship by column yields
\begin{gather}\label{Q45}
\begin{gathered}
\mu_{\pm,1}(z)=\sigma\mu^{*}_{\pm,2}(z^{*}),\quad
\mu_{\pm,2}(z)=-\sigma\mu^{*}_{\pm,1}(z^{*}),\\
\mu_{\pm,1}(z)=\frac{q_{\pm}}{z}\mu_{\pm,2}\left(-\frac{q^{2}_{0}}{z}\right), \quad
\mu_{\pm,2}(z)=-\frac{q^{*}_{\pm}}{z}\mu_{\pm,1}\left(-\frac{q^{2}_{0}}{z}\right).
\end{gathered}
\end{gather}

Because of the symmetry of $\mu_{\pm}(z)$ and the relationship between $S(z)$ and $\mu_{\pm}(z)$, we get
\begin{align}
-\sigma S^{*}(z^{*})\sigma=S(z).
\end{align}
Applying the same way, the following relation also can be obtained as
\begin{align}
S(z)=\left(\sigma_{3}Q_{-}\right)^{-1}S\left(-q^{2}_{0}/z\right)\sigma_{3}Q_{+}.
\end{align}
By column, we have
\begin{gather*}
s_{22}(z)=s^{*}_{11}(z^{*}), \quad
s_{12}(z)=-s^{*}_{21}(z^{*}),\\
s_{11}(z)=\frac{q_{+}}{q_{-}}s_{22}\left(-\frac{q^{2}_{0}}{z}\right), \quad
s_{12}(z)=\frac{q^{*}_{+}}{q_{-}}s_{21}\left(-\frac{q^{2}_{0}}{z}\right).
\end{gather*}

Furthermore, the reflection coefficient has the following symmetry
\begin{align}
\rho(z)=-\tilde{\rho}^{*}(z^{*})=
-\frac{q_{-}}{q^{*}_{-}}\tilde{\rho}\left(-\frac{q^{2}_{0}}{z}\right).
\end{align}

\subsection{Asymptotic behavior of $\mu_{\pm}$ and scattering matrix}

In this subsection, we give the asymptotic properties of $\mu_{\pm}(x,t;z)$ and $S(z)$ as follows.

\begin{prop}
The asymptotic properties of $\mu_{\pm}(x,t;z)$ and $S(z)$ are
\begin{align}
\mu_{\pm}(x,t;z)=
\left\{\begin{aligned}
&\mathbb{I}+\frac{i}{z}\sigma_{3}Q+O\left(\frac{1}{z^{2}}\right), \quad z\rightarrow\infty,\\
&\frac{i}{z}\sigma_{3}Q_{\pm}+O(1), \quad z\rightarrow0,
\end{aligned}\right.
\end{align}
and
\begin{align}
S(z)=
\left\{\begin{aligned}
&\mathbb{I}+O\left(\frac{1}{z}\right), \quad z\rightarrow\infty,\\
&diag\left(\frac{q_{+}}{q_{-}},\frac{q_{-}}{q_{+}}\right)+O(z), \quad z\rightarrow0.
\end{aligned}\right.
\end{align}
\end{prop}

\begin{proof}
When $z=\infty$, the asymptotic expansion of $\mu_{\pm}$ is
\begin{align}
\mu_{\pm}=\mu^{0}_{\pm}+\frac{\mu^{1}_{\pm}}{z}+\cdots.
\end{align}
Combining
\begin{align}
Y^{-1}_{\pm}=\frac{1}{\gamma}\left(\mathbb{I}-\frac{i}{z}Q_{\pm}\sigma_{3}\right),
\end{align}
we reconsider the equivalent Lax pair and compare the coefficients of the power term of $z$ to obtain
\begin{align*}
\mu^{0}_{\pm}=\mathbb{I},\quad \mu^{1}_{\pm}=i\sigma_{3}Q.
\end{align*}
Thus, we have
\begin{align}
\mu_{\pm}(x,t;z)=
&\mathbb{I}+\frac{i}{z}\sigma_{3}Q+O\left(\frac{1}{z^{2}}\right), \quad z\rightarrow\infty.
\end{align}
In the same way, we also can obtain the asymptotic properties of $\mu_{\pm}(x,t;z)$ when $z\rightarrow0$.

According to the asymptotic properties and relationship with the scattering matrix of $\mu_{\pm}(x,t;z)$, the asymptotic properties of $S(z)$ can be derived.
\end{proof}

\section{Generalized Riemann-Hilbert problem}

In this section, we will construct a generalized Riemann-Hilbert problem based on the previous results.

According to \eqref{Q31}, one have
\begin{align}
&\mu_{+,1}=s_{11}\mu_{-,1}+s_{21}e^{-2i\theta}\mu_{-,2},\\
&\mu_{+,2}=s_{12}e^{2i\theta}\mu_{-,1}+s_{22}\mu_{-,2}.
\end{align}

According to the analysis of the Jost function $\mu_{\pm}$ and the scattering matrix, defining a sectionally meromorphic matrix as
\begin{align}
M(x,t;z)=\left\{\begin{aligned}
M^{+}(x,t;z)=\left(\frac{\mu_{+,1}(x,t;z)}{s_{11}(z)},\mu_{-,2}(x,t;z)\right),\\
M^{-}(x,t;z)=\left(\mu_{-,1}(x,t;z),\frac{\mu_{+,2}(x,t;z)}{s_{22}(z)}\right).
\end{aligned}\right.
\end{align}
Thus, by the asymptotic of $\mu_{\pm}$ and the scattering coefficient, we know
\begin{align}
&M^{\pm}(x,t;z)\sim\mathbb{I}+O\left(\frac{1}{z}\right), \quad z\rightarrow\infty,\\
&M^{\pm}(x,t;z)\sim\frac{i}{z}\sigma_{3}Q_{-}+O(1), \quad z\rightarrow0.
\end{align}

So far, the following theorem can be concluded.
\begin{defn}
The generalized Riemann-Hilbert problem is constructed as\\
$\bullet$ $M(x,t;z)$ is meromorphic in $\mathbb{C}\backslash\Sigma$;\\
$\bullet$ The jump condition is
\begin{align}\label{3.5}
M^{-}(x,t;z)=M^{+}(x,t;z)\left(\mathbb{I}-G(x,t;z)\right),
\end{align}
 where
\begin{align*}
G(x,t;z)=\left(\begin{array}{cc}
           0 & -\tilde{\rho}(z)e^{2i\theta(x,t;z)} \\
           \rho(z)e^{2i\theta(x,t;z)} & \rho(z)\tilde{\rho}
         \end{array}\right);
\end{align*}
$\bullet$ $M^{\pm}(x,t;z)\sim\mathbb{I}+O\left(\frac{1}{z}\right), \quad z\rightarrow\infty$;\\
$\bullet$ $M^{\pm}(x,t;z)\sim\frac{i}{z}\sigma_{3}Q_{-}+O(1), \quad z\rightarrow0$.
\end{defn}

\section{The inverse scattering transform with the simple poles}

In this section, we will discuss and solve the soliton solution under the simple poles.

\subsection{Discrete spectrum and residue condition}

The discrete spectrum of the scattering problem is composed of all values $z\in\mathbb{C}\backslash\Sigma$ satisfying eigenfunctions in $L^{2}(\mathbb{R})$. These values make $s_{11}(z)=0$ for $z\in D^{+}$ and $s_{22}(z)=0$ for $z\in D^{-}$ respectively.
Assuming $z_{n}\in D^{+}\bigcap\left\{z\in\mathbb{C}:Imz>0\right\}$ is the simple poles of $s_{11}(z)$, then $s_{11}(z)=0$ but $s_{11}(z)\neq0$ for $n=1,2,\cdots,N$.

From the symmetry of the scattering data, one obtains
\begin{align}\label{4.1}
s_{11}(z_{n})=0\Leftrightarrow s_{22}(z^{*}_{n})=0\Leftrightarrow s_{22}(-q^{2}_{0}/z_{n})=0
\Leftrightarrow s_{11}(-q^{2}_{0}/z^{*}_{n})=0,
\end{align}
which yields a quartet of discrete eigenvalues, namely,
\begin{align}\label{4.2}
\mathbb{Z}=\left\{z_{n}, z^{*}_{n}, -q^{2}_{0}/z_{n}, -q^{2}_{0}/z^{*}_{n}\right\}, \quad n=1,2,\cdots,N.
\end{align}

According to \eqref{Q30}, When $s_{11}(z_{n})=0$, there is a constant $b_{n}$ and $\tilde{b}_{n}$ that do not depend on $x$, $t$ and $z$ such that
\addtocounter {equation}{1}
\begin{align}
&\psi_{+,1}(x,t;z_{n})=b_{n}\psi_{-,2}(x,t;z_{n}),\tag{\theequation a}\label{4.3a}\\
&\psi_{+,2}(x,t;z^{*}_{n})=\tilde{b}_{n}\psi_{-,1}(x,t;z^{*}_{n}).\tag{\theequation b}\label{4.3b}
\end{align}

The transformation \eqref{Q18} yields
\begin{align}\label{4.4}
\left\{\begin{aligned}
&\psi_{\pm,1}=\mu_{\pm,1}e^{i\theta(z)},\\
&\psi_{\pm,2}=\mu_{\pm,2}e^{-i\theta(z)},
\end{aligned}\right.
\end{align}
then \eqref{4.3a} and \eqref{4.3b} derive
\begin{align}
&\mu_{+,1}(z_{n})=b_{n}e^{-2i\theta(z_{n})}\mu_{-,2}(z_{n}),\label{4.5a}\\
&\mu_{+,2}(z^{*}_{n})=\tilde{b}_{n}e^{2i\theta(z^{*}_{n})}\mu_{-,1}(z^{*}_{n}).\label{4.5b}
\end{align}

Therefore, we get the residue condition as
\begin{align}
&\mathop{Res}_{z=z_{n}}\left[\frac{\mu_{+,1}(z)}{s_{11}(z)}\right]
=\frac{\mu_{+,1}(z_{n})}{s'_{11}(z)}
=\frac{b_{n}}{s'_{11}(z)}e^{-2i\theta(z_{n})}\mu_{-,2}(z_{n}),\label{4.6a}\\
&\mathop{Res}_{z=z^{*}_{n}}\left[\frac{\mu_{+,2}(z)}{s_{22}(z)}\right]
=\frac{\mu_{+,2}(z^{*}_{n})}{s'_{22}(z^{*})}
=\frac{\tilde{b}_{n}}{s'_{22}(z^{*})}e^{2i\theta(z^{*}_{n})}\mu_{-,1}(z^{*}_{n}).\label{4.6b}
\end{align}
For the convenience of calculation, introducing
\begin{align*}
C_{n}=\frac{b_{n}}{s'_{11}(z)},\quad \tilde{C}_{n}=\frac{\tilde{b}_{n}}{s'_{22}(z^{*})},
\end{align*}
then \eqref{4.6a} and \eqref{4.6b} become
\begin{align}
&\mathop{Res}_{z=z_{n}}\left[\frac{\mu_{+,1}(z)}{s_{11}(z)}\right]
=C_{n}e^{-2i\theta(z_{n})}\mu_{-,2}(z_{n}),\label{4.7a}\\
&\mathop{Res}_{z=z^{*}_{n}}\left[\frac{\mu_{+,2}(z)}{s_{22}(z)}\right]
=\tilde{C}_{n}e^{2i\theta(z^{*}_{n})}\mu_{-,1}(z^{*}_{n}).\label{4.7b}
\end{align}

Doing some simple calculations based on symmetry after substituting \eqref{4.4} into \eqref{4.3a} and \eqref{4.3b}, it is easy to get the relation between $b_{n}$ and $\tilde{b}_{n}$ by compare the results, i.e.,
\begin{align}\label{4.8}
-b^{*}_{n}=\tilde{b}_{n}.
\end{align}
Combining $s'_{11}(z_{n})=\left(s'_{22}(z^{*}_{n})\right)^{*}$, we have
\begin{align}\label{4.9}
\tilde{C}_{n}=-C^{*}_{n}.
\end{align}

Applying the symmetry of $\mu_{\pm}$ for
$\mu_{+,1}(z_{n})=b_{n}\mu_{-,2}(z_{n})e^{-2i\theta(z_{n})}$, it can be see
\begin{align}\label{4.10}
\psi_{+,2}\left(-\frac{q^{2}_{0}}{z_{n}}\right)
=-\frac{q^{*}_{-}}{q_{+}}b_{n}\psi_{-,1}\left(-\frac{q^{2}_{0}}{z_{n}}\right).
\end{align}
Similarly,
\begin{align}\label{4.11}
\psi_{+,1}\left(-\frac{q^{2}_{0}}{z^{*}_{n}}\right)
=-\frac{q_{-}}{q^{*}_{+}}\tilde{b}_{n}\psi_{-,2}\left(-\frac{q^{2}_{0}}{z^{*}_{n}}\right).
\end{align}
Next, we need to deal with $s_{11}(z)=\frac{q_{+}}{q_{-}}s_{22}\left(-\frac{q^{2}_{0}}{z}\right)$ as follows. Taking the derivative of $z$, substitute $z=z^{*}$, and making the conjugate on both sides, we can get
\begin{align}\label{4.12}
s'_{11}\left(-\frac{q^{2}_{0}}{z^{*}_{n}}\right)=
\left(\frac{q_{-}}{q_{+}}\right)^{*}\left(\frac{z_{n}}{q_{0}}\right)^{2}
\left(s'_{11}(z_{n})\right)^{*}.
\end{align}
In the same way, one obtains
\begin{align}\label{4.13}
s'_{22}\left(-\frac{q^{2}_{0}}{z_{n}}\right)=
\frac{q_{-}}{q_{+}}\left(\frac{z_{n}}{q_{0}}\right)^{2}\left(s'_{22}(z^{*}_{n})\right)^{*}.
\end{align}
Therefore, the residue conditions on $z=-\frac{q^{2}_{0}}{z^{*}_{n}}$ and $z=-\frac{q^{2}_{0}}{z_{n}}$ are
\begin{align}\label{4.14a}
\begin{aligned}
\mathop{Res}_{z=-\frac{q^{2}_{0}}{z^{*}_{n}}}\left[\frac{\mu_{+,1}(z)}{s_{11}(z)}\right]
&=-\frac{q_{-}}{q^{*}_{-}}\left(\frac{z_{n}}{q_{0}}\right)^{2}\tilde{C}_{n}
e^{-2i\theta(-q^{2}_{0}/z^{*}_{n})}\mu_{-,2}\left(-\frac{q^{2}_{0}}{z^{*}_{n}}\right)\\
&=C_{N+n}e^{-2i\theta(-q^{2}_{0}/z^{*}_{n})}\mu_{-,2}\left(-\frac{q^{2}_{0}}{z^{*}_{n}}\right),\\
\end{aligned}
\end{align}
\begin{align}\label{4.14b}
\begin{aligned}
\mathop{Res}_{z=-\frac{q^{2}_{0}}{z_{n}}}\left[\frac{\mu_{+,2}(z)}{s_{22}(z)}\right]
&=-\frac{q^{*}_{-}}{q_{-}}\left(\frac{z_{n}}{q_{0}}\right)^{2}C_{n}
e^{2i\theta(-q^{2}_{0}/z_{n})}\mu_{-,1}\left(-\frac{q^{2}_{0}}{z_{n}}\right)\\
&=\tilde{C}_{N+n}e^{-2i\theta(-q^{2}_{0}/z_{n})}\mu_{-,1}\left(-\frac{q^{2}_{0}}{z_{n}}\right),
\end{aligned}
\end{align}
where
\begin{align*}
C_{N+n}=-\frac{q_{-}}{q^{*}_{-}}\left(\frac{z_{n}}{q_{0}}\right)^{2}\tilde{C}_{n},\quad
\tilde{C}_{N+n}=-\frac{q^{*}_{-}}{q_{-}}\left(\frac{z_{n}}{q_{0}}\right)^{2}C_{n}.
\end{align*}

For convenience, numbering a quartet of discrete eigenvalues uniformly and defining $\xi_{n}=z_{n}$ and $\xi_{N+n}=-\frac{q^{2}_{0}}{z_{n}}$ for $n=1,2,\cdots,N$, the residue condition of $M^{+}(z)$ and $M^{-}(z)$ are
\begin{align}
&\mathop{Res}_{z=\xi_{n}}M^{+}(z)=\left(C_{n}e^{-2i\theta(\xi_{n})}\mu_{-,2}(\xi_{n}), 0\right),\\
&\mathop{Res}_{z=\xi^{*}_{n}}M^{-}(z)=\left(0, \tilde{C}_{n}e^{2i\theta(\xi^{*}_{n})}\mu_{-,1}(\xi^{*}_{n})\right),
\end{align}
here $n=1,2,\cdots,2N$.

\subsection{Reconstruct the formula for potential}

To solve the Riemann-Hilbert(RH) problem, we need to transform the original RH problem into a regular RH problem by subtracting the asymptotic properties and the pole contributions. Thus \eqref{3.5} can be rewritten as
\begin{gather}\label{4.16}
\begin{aligned}
&M^{-}(z)-\mathbb{I}-\frac{i}{z}\sigma_{3}Q_{-}-\sum^{2N}_{n=1}
\frac{\mathop{Res}\limits_{z=\xi^{*}_{n}}M^{-}(z)}{z-\xi^{*}_{n}}
-\sum^{2N}_{n=1}\frac{\mathop{Res}\limits_{z=\xi_{n}}M^{+}(z)}{z-\xi_{n}}\\
&=M^{+}(z)-\mathbb{I}-\frac{i}{z}\sigma_{3}Q_{-}-\sum^{2N}_{n=1}
\frac{\mathop{Res}\limits_{z=\xi^{*}_{n}}M^{-}(z)}{z-\xi^{*}_{n}}
-\sum^{2N}_{n=1}\frac{\mathop{Res}\limits_{z=\xi_{n}}M^{+}(z)}{z-\xi_{n}}-M^{+}(z)G(z).
\end{aligned}
\end{gather}
For \eqref{4.16}, we know that the first five terms on the left side are analyzed on $D_{-}$, while the first five terms on the right side are analyzed on $D_{+}$.
Introducing the projection operators $P_{\pm}$ on $\Sigma$ as
\begin{align}\label{4.17}
P_{\pm}[f](z)=\frac{1}{2\pi i}\int_{\Sigma}\frac{f(\zeta)}{\zeta-(z\pm i0)}d\zeta,
\end{align}
where $\int_{\Sigma}$ means the integral along the oriented contour shown in Fig.1 and $z\pm i0$
is the limit which is taken from the left/right of $z(z\in\Sigma)$ respectively.
Via applying Plemelj's formulae, \eqref{4.16} can be solved as
\begin{align}\label{4.18}
\begin{aligned}
M(z)=&\mathbb{I}+\frac{i}{z}\sigma_{3}Q_{-}
+\sum^{2N}_{n=1}\frac{\mathop{Res}_{z=\xi^{*}_{n}}M^{-}(z)}{z-\xi^{*}_{n}}
+\sum^{2N}_{n=1}\frac{\mathop{Res}_{z=\xi_{n}}M^{+}(z)}{z-\xi_{n}}\\
&-\frac{1}{2\pi i}\int_{\Sigma}\frac{M^{+}(s)G(s)}{s-z}ds,\quad z\in\mathbb{C}\setminus\Sigma.
\end{aligned}
\end{align}
Considering the residue condition, the second column of \eqref{4.18}, in $z=z_{n}$ and $z=-q^{2}_{0}/z^{*}_{n}$, i.e., $z=\xi_{n}$, is solved as
\begin{align}\label{4.19a}
\mu_{-,2}(\xi_{n})=\left(\begin{array}{c}
                     -q^{*}_{-}/\xi_{n} \\
                     1
                   \end{array}\right)
+\sum^{2N}_{k=1}\frac{\tilde{C}_{k}e^{2i\theta(\xi^{*}_{k})}}{\xi_{n}-\xi^{*}_{k}}
\mu_{-,1}(\xi^{*}_{k})+\frac{1}{2\pi i}\int_{\Sigma}\frac{(M^{+}G)_{2}(s)}{s-\xi_{n}}ds.
\end{align}
Similarly, the first column of \eqref{4.18} in $z=z^{*}_{n}$ and $z=-q^{2}_{0}/z_{n}$, i.e., $z=\xi^{*}_{n}$, is
\begin{align}\label{4.19b}
\mu_{-,1}(\xi^{*}_{n})=\left(\begin{array}{c}
                     1\\
                     q_{-}/\xi^{*}_{n}
                   \end{array}\right)
+\sum^{2N}_{j=1}\frac{C_{j}e^{-2i\theta(\xi_{j})}}{\xi^{*}_{n}-\xi_{j}}
\mu_{-,2}(\xi_{j})+\frac{1}{2\pi i}\int_{\Sigma}\frac{(M^{+}G)_{1}(s)}{s-\xi^{*}_{n}}ds.
\end{align}
When $z\rightarrow\infty$, asymptotic expansion of \eqref{4.18} is
\begin{align}\label{4.20}
\begin{aligned}
M(z)=\mathbb{I}+\frac{1}{z}&\left\{
i\sigma_{3}Q_{-}
+\sum^{2N}_{n=1}\left(\mathop{Res}_{z=\xi^{*}_{n}}M^{-}(z)+\mathop{Res}_{z=\xi_{n}}M^{+}(z)\right)\right.\\
&\left.-\frac{1}{2\pi i}\int_{\Sigma}M^{+}(s)G(s)ds\right\}+O\left(\frac{1}{z^{2}}\right).
\end{aligned}
\end{align}
Taking $M=M^{-}$, we obtain the follow result by combining the (2, 1)-element of asymptotic behavior of $\mu_{\pm}$ as
\begin{align}\label{4.21}
q=q_{-}+\sum^{2N}_{n=1}C_{n}e^{-2i\theta(\xi_{n})}\mu_{-,22}(\xi_{n})-
\frac{1}{2\pi i}\int_{\Sigma}(M^{+}(s)G(s))_{21}ds.
\end{align}

\subsection{Trace formulate and theta condition}
Recall that $s_{11}$ and $s_{22}$ are analytic on $D^{+}$ and $D^{-}$ respectively, and the discrete spectrum is $\mathbb{Z}=\{z_{n}, z^{*}_{n}, -q^{2}_{0}/z_{n}, -q^{2}_{0}/z^{*}_{n}\}$ for $n=1,2,\cdots,N$. Constructing the follow functions
\begin{align}\label{4.22}
\begin{split}
\beta^{+}_{1}(z)=s_{11}(z)\prod^{N}_{n=1}\frac{(z-z^{*}_{n})(z+q^{2}_{0}/z_{n})}
{(z-z_{n})(z+q^{2}_{0}/z^{*}_{n})},\\
\beta^{-}_{1}(z)=s_{22}(z)\prod^{N}_{n=1}\frac{(z-z_{n})(z+q^{2}_{0}/z^{*}_{n})}
{(z-z^{*}_{n})(z+q^{2}_{0}/z_{n})},
\end{split}
\end{align}
here $\beta^{+}_{1}(z)$ and $\beta^{-}_{1}(z)$ are analytic in $D^{+}$ and $D^{-}$ respectively, and there is no zero point.

In addition, $\beta^{\pm}_{1}(z)\rightarrow1$ as $z\rightarrow\infty$ via the asymptotic behavior of the scattering data.

For all $z\in\Sigma$, we have $\beta^{+}_{1}(z)\beta^{-}_{1}(z)=s_{11}(z)s_{22}(z)$. According to $detS(z)=1$, it can be seen that
\begin{align*}
\frac{1}{s_{11}(z)s_{22}(z)}=1-\rho(z)\tilde{\rho}(z)=1+\rho(z)\rho^{*}(z^{*}),
\end{align*}
then
\begin{align}\label{4.23}
\beta^{+}_{1}(z)\beta^{-}_{1}(z)=\frac{1}{1+\rho(z)\rho^{*}(z^{*})},\quad z\in\Sigma.
\end{align}
By taking the logarithm of both sides of the above formula and using the Plemelj's formulae we know
\begin{align}\label{4.24}
\log\beta^{\pm}_{1}(z)=\mp\frac{1}{2\pi i}\int_{\Sigma}\frac{\log[1+\rho(\zeta)\rho^{*}(\zeta^{*})]}
{\zeta-z}d\zeta, \quad z\in D^{\pm}.
\end{align}

Substituting \eqref{4.24} into \eqref{4.22}, we obtain trace formulate as
\begin{align}\label{4.25}
\begin{split}
&s_{11}(z)=exp\left[-\frac{1}{2\pi i}\int_{\Sigma}\frac{\log[1+\rho(\zeta)\rho^{*}(\zeta^{*})]}
{\zeta-z}d\zeta\right]\prod^{N}_{n=1}\frac{(z-z_{n})(z+q^{2}_{0}/z^{*}_{n})}
{(z-z^{*}_{n})(z+q^{2}_{0}/z_{n})},\quad z\in D^{+},\\
&s_{22}(z)=exp\left[\frac{1}{2\pi i}\int_{\Sigma}\frac{\log[1+\rho(\zeta)\rho^{*}(\zeta^{*})]}
{\zeta-z}d\zeta\right]\prod^{N}_{n=1}\frac{(z-z^{*}_{n})(z+q^{2}_{0}/z_{n})}
{(z-z_{n})(z+q^{2}_{0}/z^{*}_{n})},\quad z\in D^{-}.
\end{split}
\end{align}
Taking $z\rightarrow0$ for \eqref{4.25}, theta condition can be obtained as
\begin{align}\label{4.26}
arg\left(\frac{q_{+}}{q_{-}}\right)=\frac{1}{2\pi }\int_{\Sigma}
\frac{\log[1+\rho(\zeta)\rho^{*}(\zeta^{*})]}{\zeta-z}d\zeta
+4\sum^{N}_{n=1}argz_{n}.
\end{align}

\subsection{Reflection-less potential}

Now, we study the potential $q(x,t)$ when the reflection coefficient $\rho(z)$ gradually disappears. At this point, without the jump from $M^{+}$ to $M^{-}$, the inverse problem will also degenerate into an algebraic system. According to \eqref{4.19a} and \eqref{4.19b}, it is easy to see that
\begin{align}
&\mu_{-,22}(\xi_{n})=1+\sum^{2N}_{k=1}\frac{\tilde{C}_{k}}{\xi_{n}-\xi^{*}_{k}}
e^{2i\theta(\xi^{*}_{k})}\mu_{-,21}(\xi^{*}_{k})
=1-\sum^{2N}_{k=1}c^{*}_{j}(\xi^{*}_{k})\mu_{-,21}(\xi^{*}_{k}),\label{4.27a}\\
&\mu_{-,21}(\xi^{*}_{k})=\frac{q_{-}}{\xi^{*}_{k}}+\sum^{2N}_{j=1}\frac{C_{j}}{\xi^{*}_{k}-\xi_{j}}
e^{-2i\theta(\xi_{j})}\mu_{-,22}(\xi_{j})
=\frac{q_{-}}{\xi^{*}_{k}}+\sum^{2N}_{j=1}c_{j}(\xi^{*}_{k})\mu_{-,22}(\xi_{j}),\label{4.27b}
\end{align}
here
\begin{align*}
c_{j}(z)=\frac{C_{j}}{z-\xi_{j}}e^{-2i\theta(\xi_{j})},
\end{align*}
with $j=1,2,\cdots,2N$.
Substituting \eqref{4.27b} into \eqref{4.27a} get
\begin{align}\label{4.29}
\mu_{-,22}(\xi_{n})=1-\sum^{2N}_{k=1}c^{*}_{k}(\xi^{*}_{n})\frac{q_{-}}{\xi^{*}_{k}}
-\sum^{2N}_{k=1}\sum^{2N}_{j=1}c^{*}_{k}(\xi^{*}_{n})c_{j}(\xi^{*}_{k})\mu_{-,22}(\xi_{j}).
\end{align}
Introducing
\begin{align}
&X_{n}=\mu_{-,22}(\xi_{n}),\quad X=(X_{1},\cdots,X_{2N})^{T}\\
&B_{n}=1-\sum^{2N}_{k=1}c^{*}_{k}(\xi^{*}_{n})\frac{q_{-}}{\xi^{*}_{k}},
\quad B=(B_{1},\cdots,B_{2N})^{T},\\
&A_{n,j}=\sum^{2N}_{j=1}c^{*}_{k}(\xi^{*}_{n})c_{j}(\xi^{*}_{k}),\quad
A=(A_{n,j})_{2N\times2N},
\end{align}
then \eqref{4.29} is converted into component form
\begin{align}
MX=B,
\end{align}
where $M=I+A=(M_{1},\cdots,M_{2N})$.

Using Cramer's Rule, we can calculate that
\begin{align}
X_{n}=\frac{detM^{ext}_{n}}{detM},
\end{align}
here $M^{ext}_{n}=(M_{1},\cdots,M_{n-1},B,M_{n+1},\cdots,M_{2N})$.
Thus
\begin{align}
q(x,t)=q_{-}-\frac{detM^{aug}}{detM},
\end{align}
where
\begin{align*}
&M^{aug}=\left(\begin{array}{cc}
          0 & Y^{T} \\
          B & M
        \end{array}\right),\\
&Y_{n}=C_{n}e^{-2i\theta(\xi_{n})}, Y=(Y_{1},\cdots,Y_{2N}).
\end{align*}

\subsection{Soliton solutions with the single poles}

In the last subsection, we specifically obtain the exact expression of the soliton solution of the sNLS equation and the images of the solutions by selecting appropriate parameters to study the propagation behavior and dynamic behavior of the solution.

Let $N=1$, we discuss the figures under the appropriate parameters.
\\

{\rotatebox{0}{\includegraphics[width=2.75cm,height=2.5cm,angle=0]{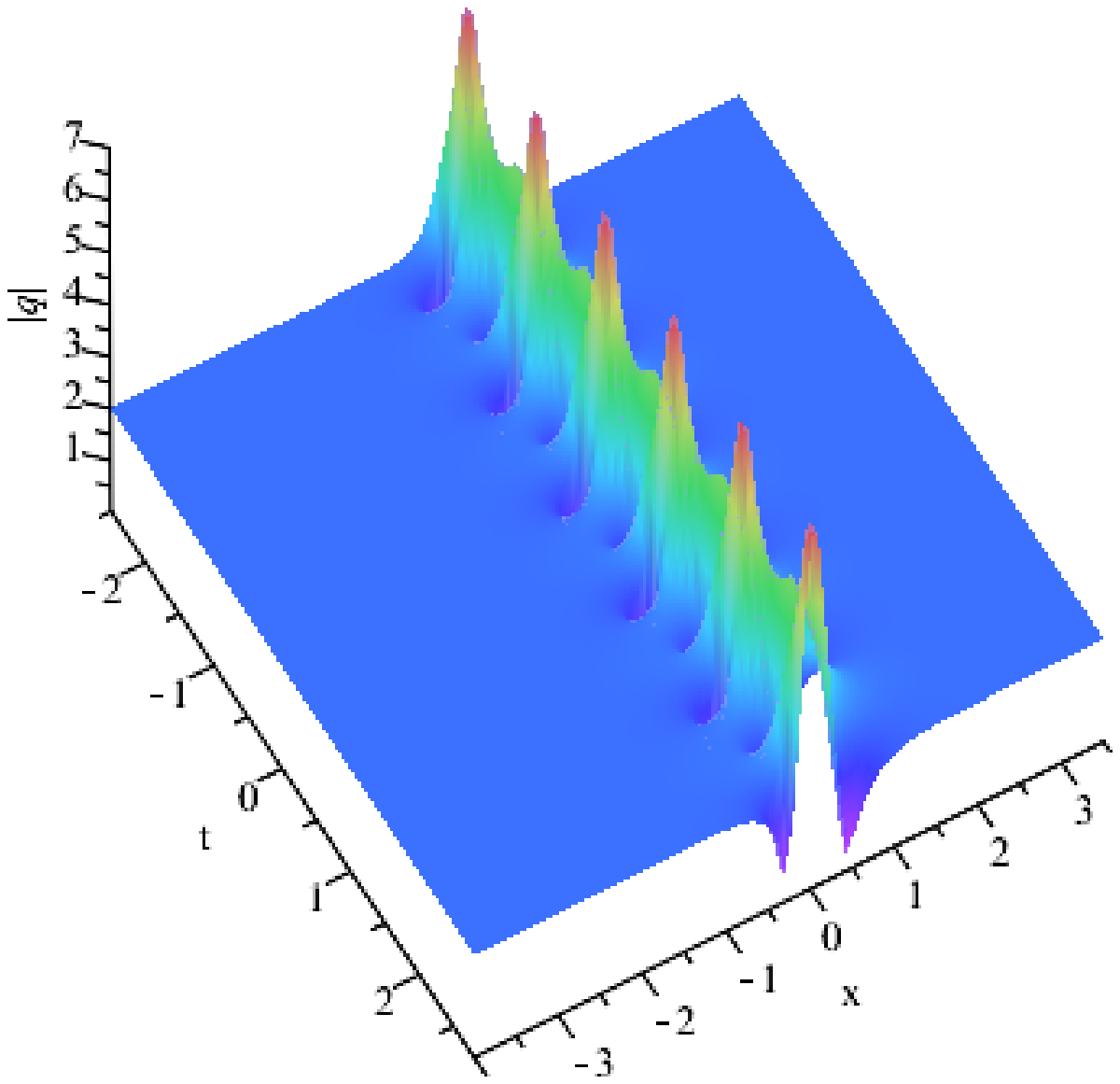}}}
{\rotatebox{0}{\includegraphics[width=2.75cm,height=2.5cm,angle=0]{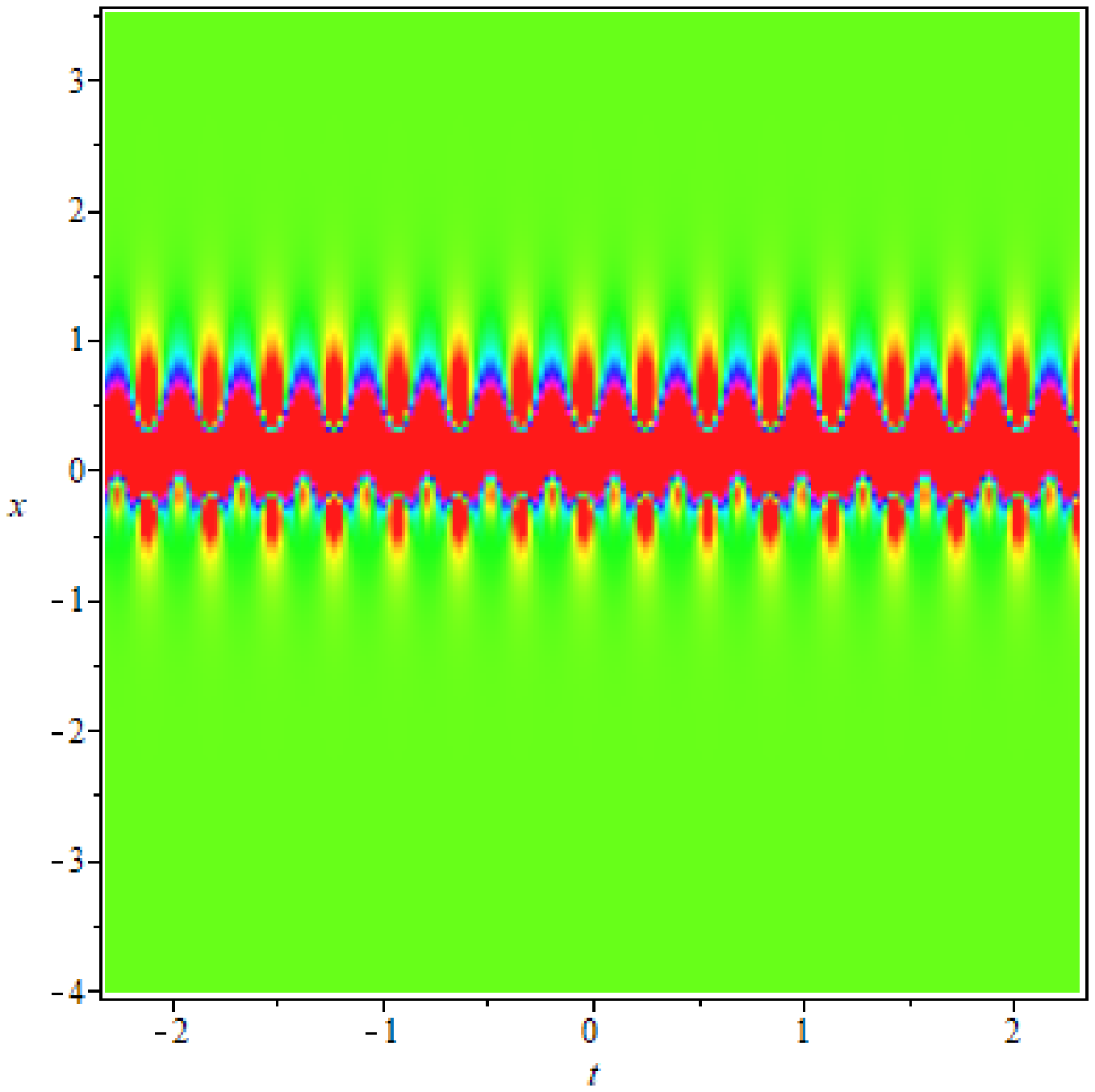}}}
{\rotatebox{0}{\includegraphics[width=2.75cm,height=2.5cm,angle=0]{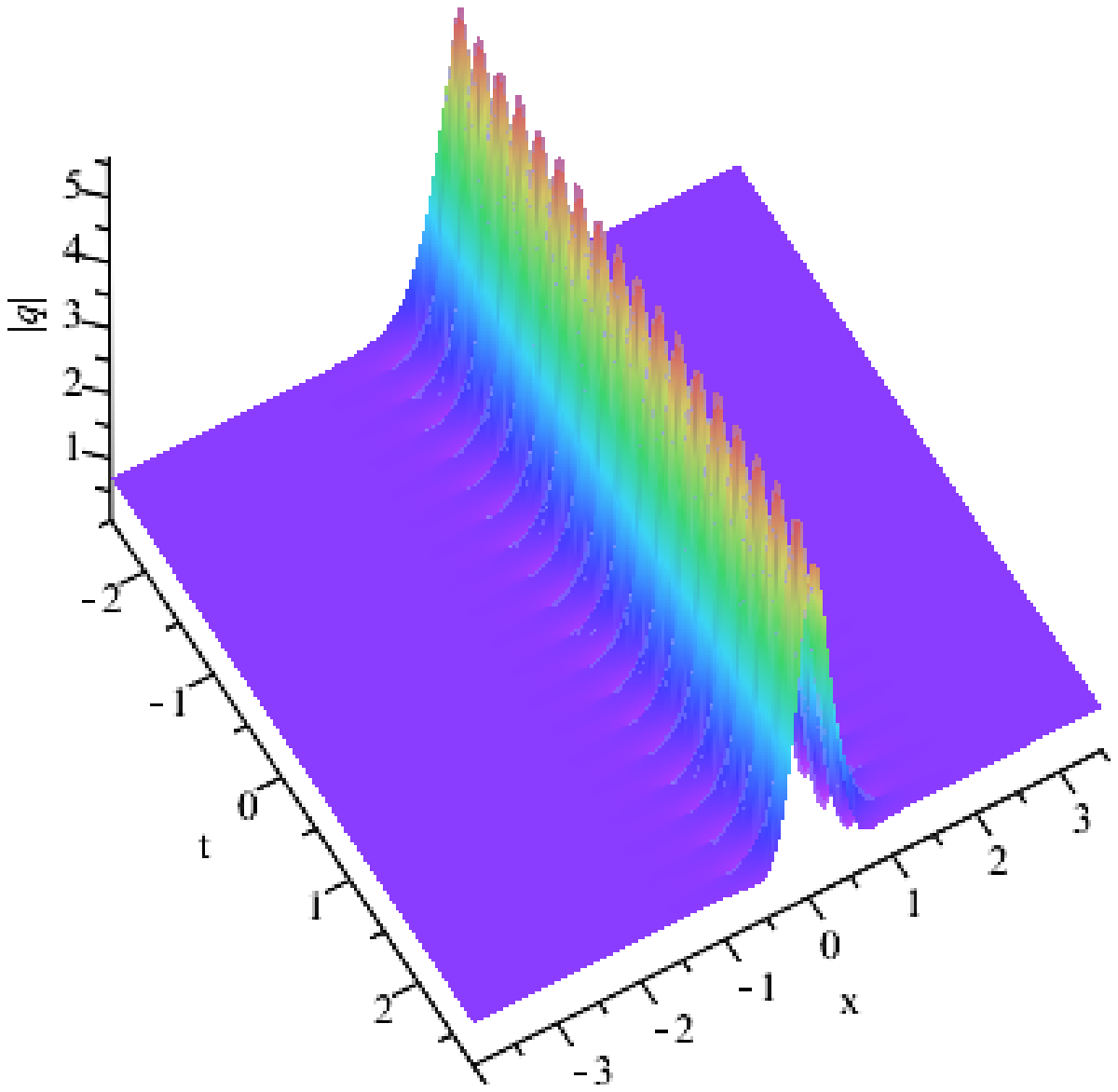}}}
{\rotatebox{0}{\includegraphics[width=2.75cm,height=2.5cm,angle=0]{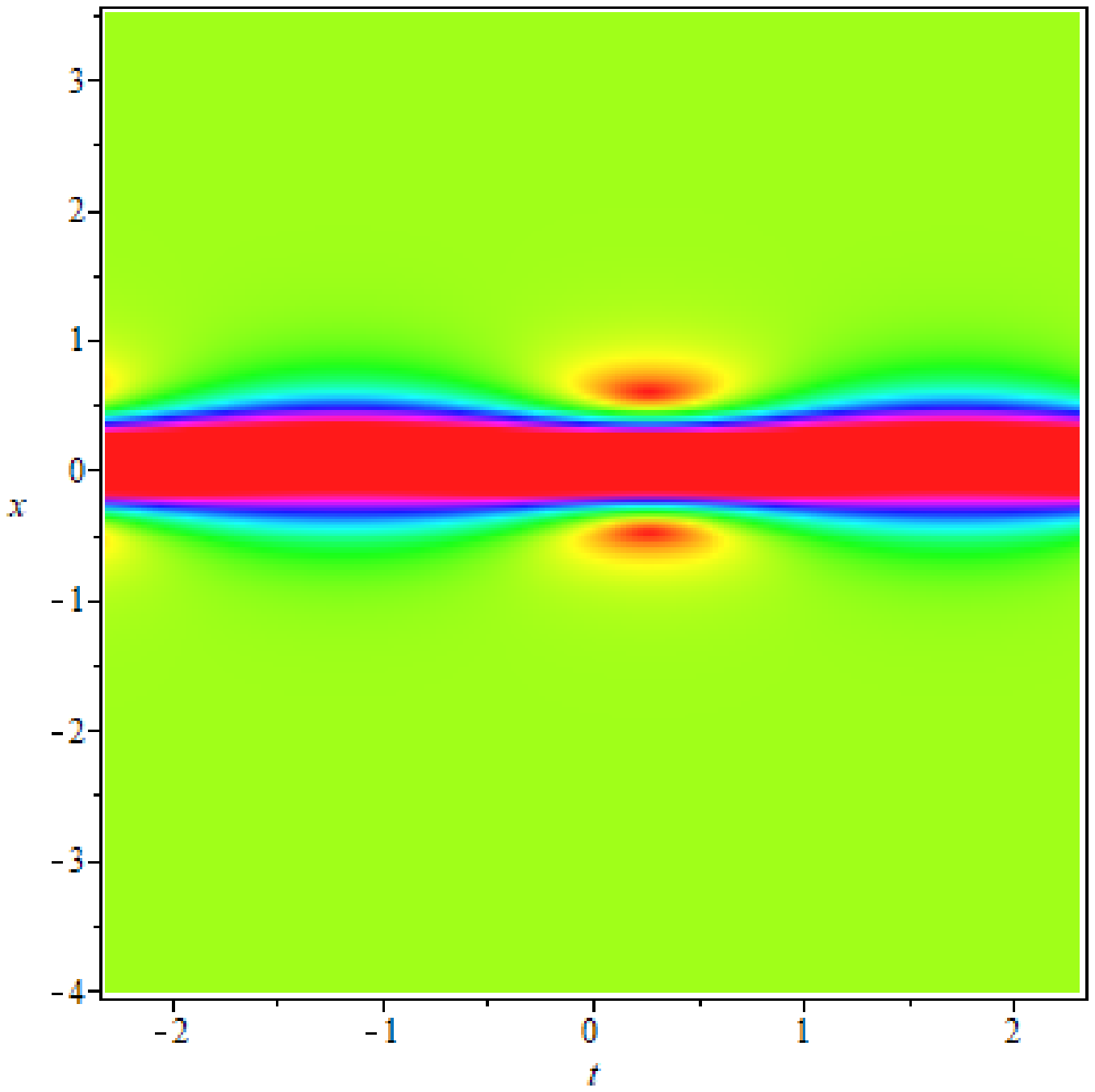}}}

 $\quad\qquad(\textbf{a})\qquad \ \quad\qquad\qquad(\textbf{b})
\ \qquad\quad\quad\quad\qquad(\textbf{c})\qquad\qquad\quad\qquad(\textbf{d})
\ \qquad$\\
\noindent { \small \textbf{Figure 2.}
The breather wave solution for $|q|$ with the parameters selection
$\delta=1, \xi_{1}=5i$.
$\textbf{(a)}$ the soliton solution with $q_{-}=2$,
$\textbf{(b)}$ the density plot corresponding to (a),
$\textbf{(c)}$ the soliton solution with $q_{-}=0.7i$,
$\textbf{(d)}$ the density plot corresponding to (c).}\\

In Fig.2, we have discussed the image when $q_{-}$ is a real number and an imaginary number, and found that the breathing wave solution can be obtained in both cases.

\noindent
{\rotatebox{0}{\includegraphics[width=3.5cm,height=3.5cm,angle=0]{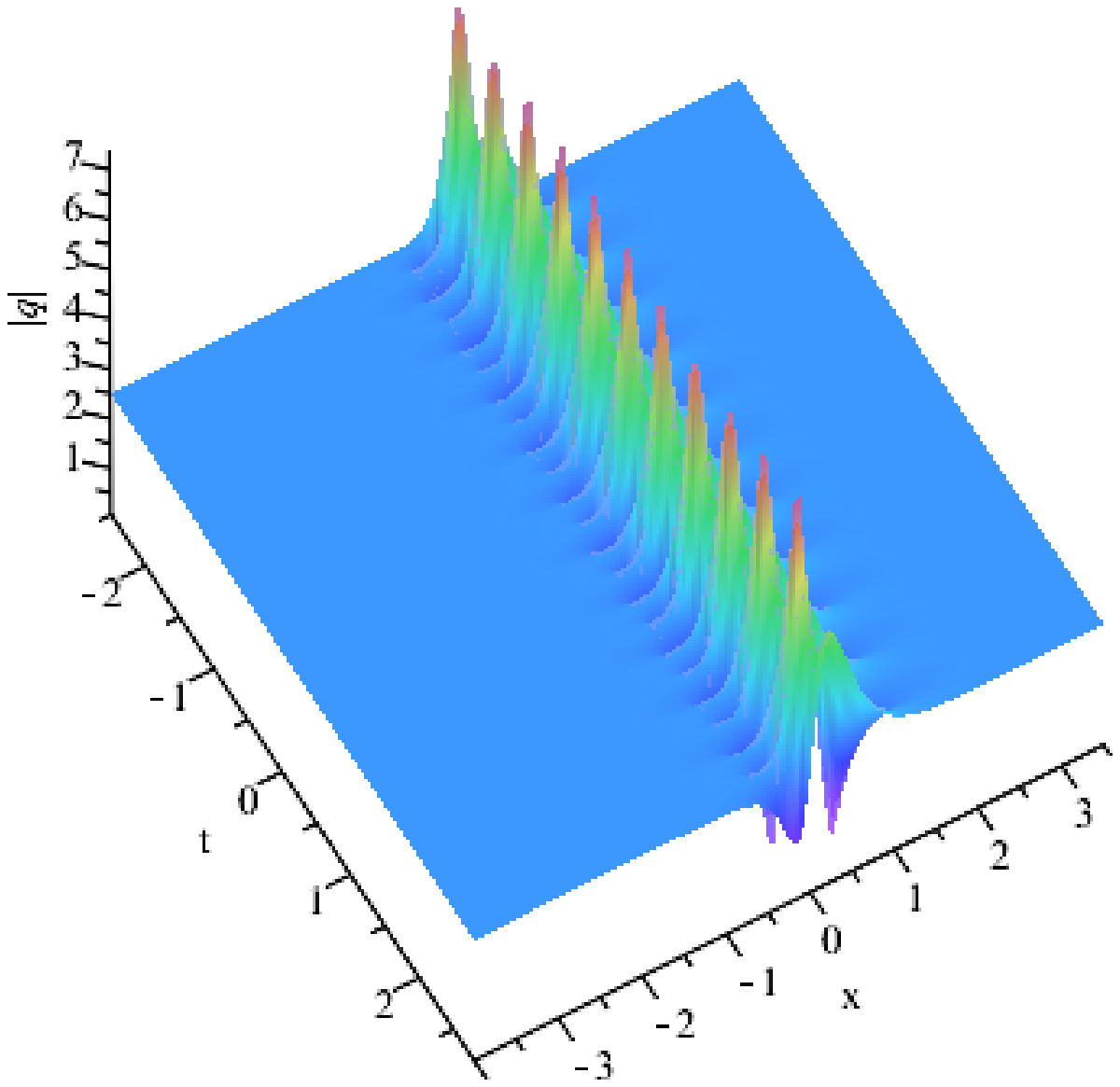}}}
\qquad\quad
{\rotatebox{0}{\includegraphics[width=3.5cm,height=3.5cm,angle=0]{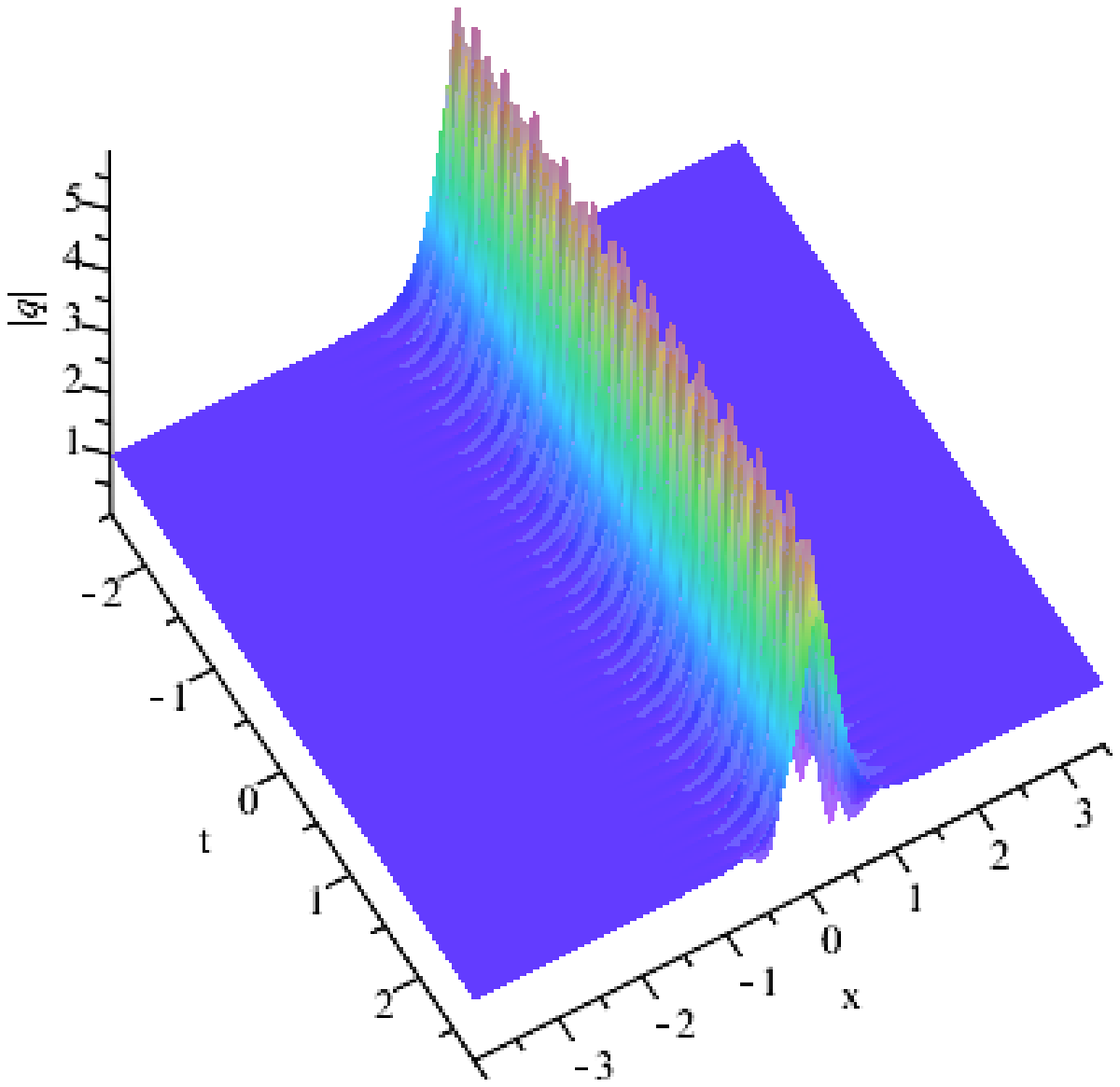}}}
\qquad~
{\rotatebox{0}{\includegraphics[width=3.5cm,height=3.5cm,angle=0]{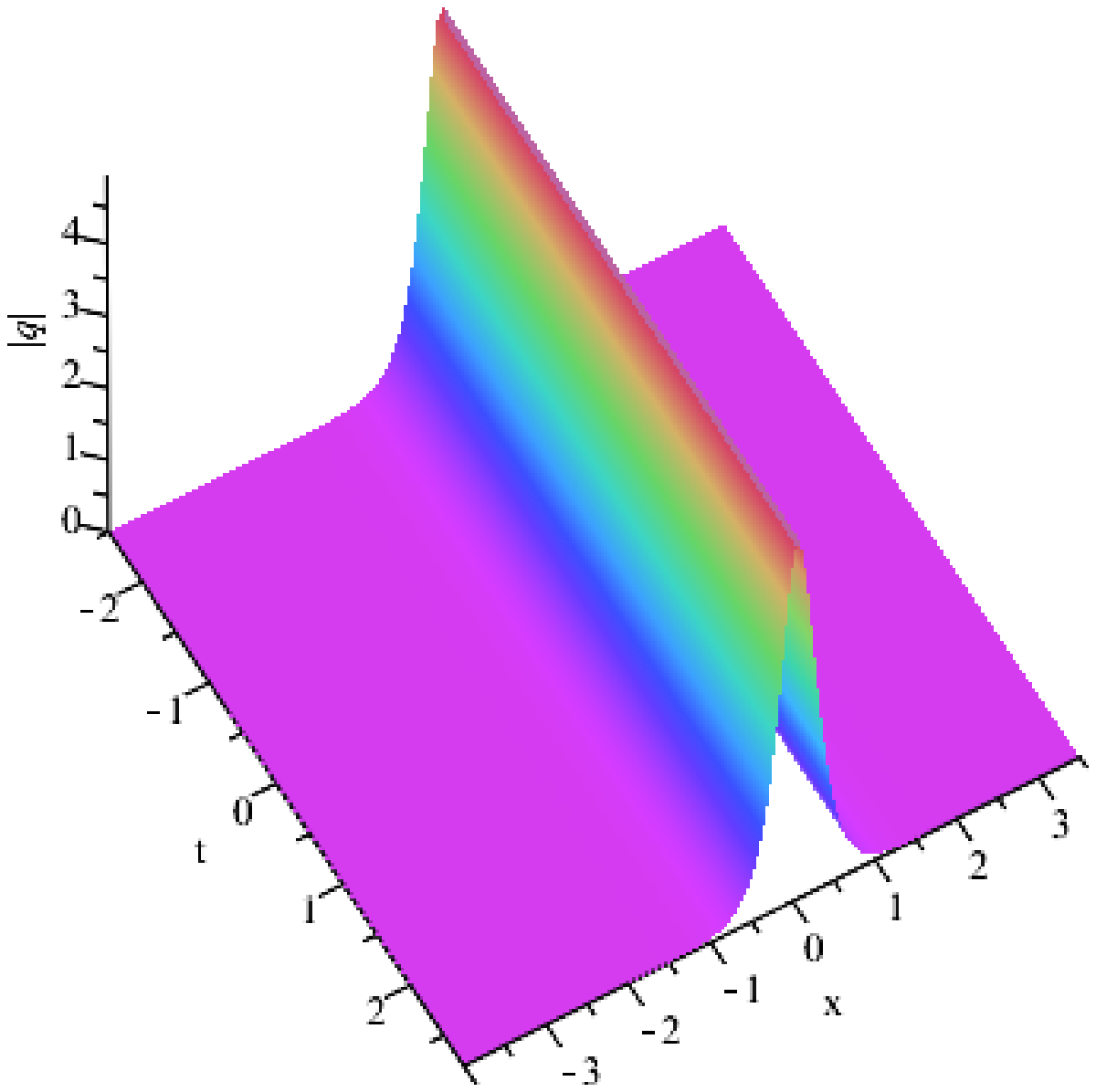}}}

\qquad\quad $(a)$
\qquad\qquad\qquad\qquad\qquad\quad\quad$(b)$\quad\qquad  \qquad\qquad\qquad\qquad  $(c)$\\
\noindent { \small \textbf{Figure 3.}
The breather wave solution for $|q|$ with the parameters selection
$\delta=1, \xi_{1}=5i$.
$\textbf{(a)}$ the soliton solution with $q_{-}=2.5$,
$\textbf{(b)}$ the soliton solution with $q_{-}=1$,
$\textbf{(c)}$ the soliton solution with $q_{-}=0.01$.}\\

In Fig.3, when the value of $q_{-}$ is different, the shape of the wave is also different. When the value of $q_{-}$ gets smaller and smaller and gradually approaches 0, the breathing wave gradually becomes a bright soliton wave.

Next, we discuss the figures of the soliton when $N=2$. In Fig.4, it is Obvious to seen that the two solitons are side by side and do not affect each other.
\\

\noindent
{\rotatebox{0}{\includegraphics[width=3.5cm,height=3.5cm,angle=0]{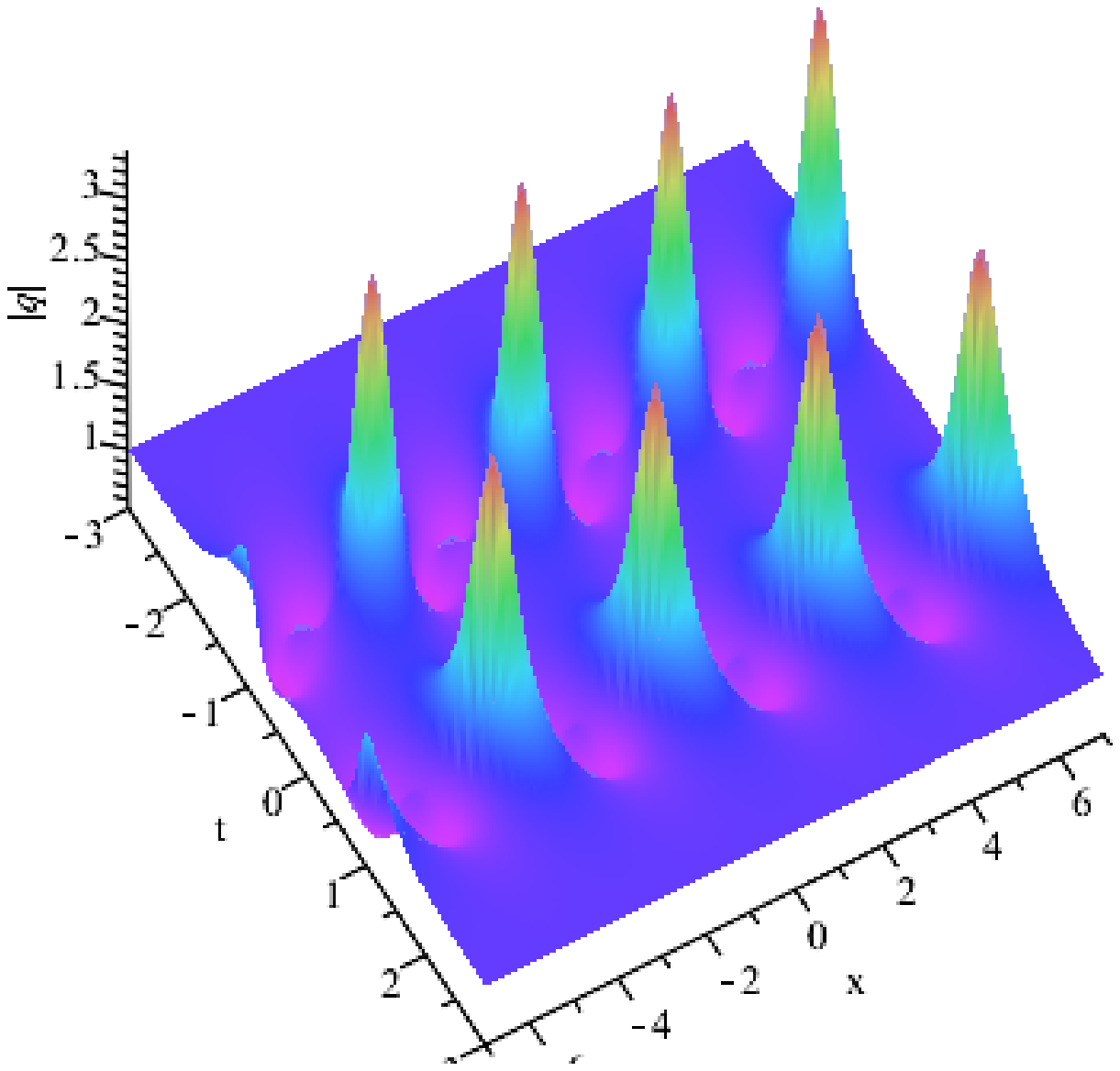}}}
\qquad\quad
{\rotatebox{0}{\includegraphics[width=3.5cm,height=3.5cm,angle=0]{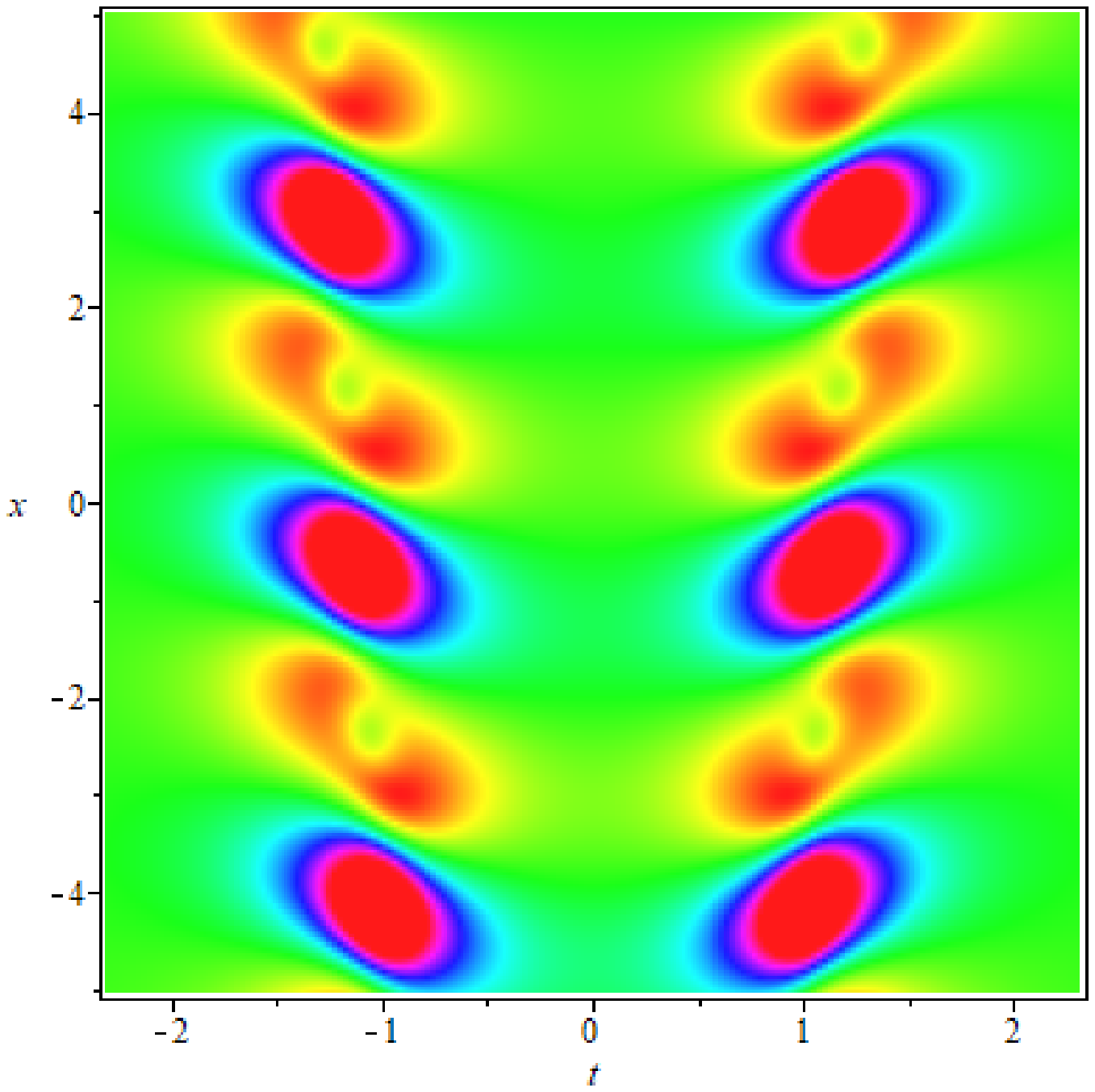}}}
\qquad~
{\rotatebox{0}{\includegraphics[width=3.5cm,height=3.5cm,angle=0]{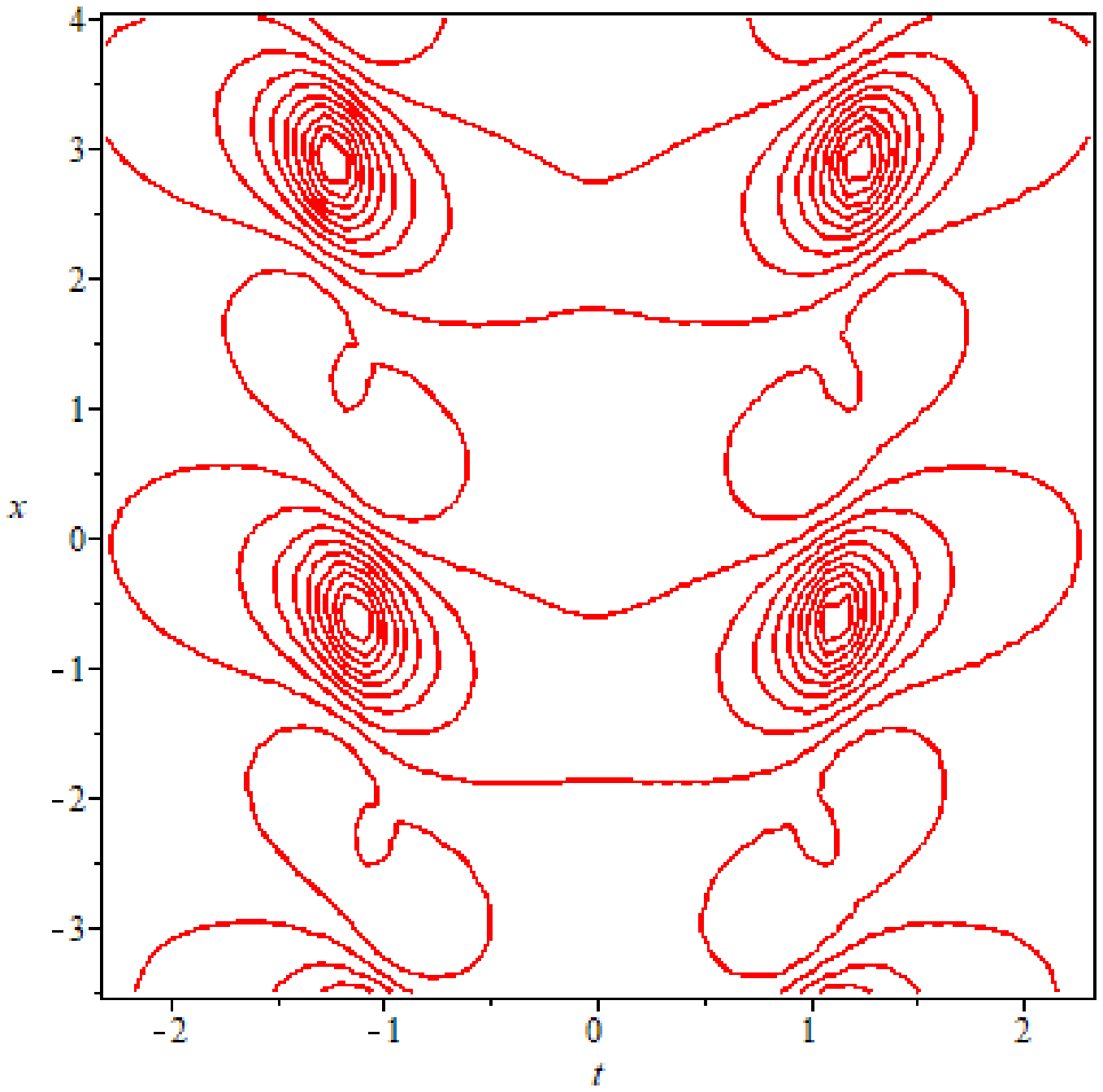}}}

\qquad\quad $(a)$
\qquad\qquad\qquad\qquad\qquad\quad\quad$(b)$\quad\qquad  \qquad\qquad\qquad\qquad  $(c)$\\
\noindent { \small \textbf{Figure 4.}
The breather wave solution for $|q|$ with the parameters selection
$\delta=0.3, \xi_{1}=1+0.5i, \xi_{2}=-1+0.5i$.
$\textbf{(a)}$ the soliton solution with $q_{-}=1$,
$\textbf{(b)}$ density plot corresponding to (a),
$\textbf{(c)}$ the contour line of the soliton solution corresponding to (a).\\

\section{The inverse scattering transform with the double poles}

In this section, we will discuss the case where $z_{n}$ is a double pole of $s_{11}$, i.e., $s_{11}(z_{n})=s'_{11}(z_{n})=0$, but $s''_{11}(z_{n})\neq0$. Therefore, our calculations in this section will be different from the simple poles case.

\subsection{Discrete spectrum and residue condition}

According to the symmetry of the scattering data $S(z)$, if $s_{11}(z_{n})=0$,
\begin{align*}
s_{22}(z^{*}_{n})=s_{22}(-q^{2}_{0}/z_{n})=s_{11}(-q^{2}_{0}/z^{*}_{n})=0.
\end{align*}
Then discrete spectrum is
\begin{align}
\mathbb{Z}=\{Z_{n},z^{*}_{n},-q^{2}_{0}/z_{n},-q^{2}_{0}/z^{*}_{n}\},\quad n=1,\cdots,N.
\end{align}
For convenience, defining $\hat{z}_{n}=-q^{2}_{0}/z_{n}$ and $\check{z}_{n}=-q^{2}_{0}/z^{*}_{n}$.

From the simple poles, we know
\begin{align}\label{5.3a}
\mu_{+,1}(z_{n})=b_{n}e^{-2i\theta_{n}}\mu_{-,2}(z_{n}).
\end{align}
Because $z_{n}$ is the double pole of $s_{11}(z)$, $s'_{11}(z_{n})=0$. For $s_{11}=\frac{1}{\gamma}W(\psi_{+,1},\psi_{-,2})$, there are
\begin{align*}
0=s'_{11}(z_{n})=\left[\frac{1}{\gamma}\left(W(\psi'_{+,1},\psi_{-,2})+W(\psi_{+,1},\psi'_{-,2})
-\frac{\gamma'}{\gamma^{2}}W(\psi_{+,1},\psi_{-,2})
\right)\right]_{z=z_{n}},
\end{align*}
i.e.,
\begin{align*}
0=W(\psi'_{+,1}(z_{n})-b_{n}\psi'_{-,2}(z_{n}),\psi_{-,2}(z_{n})).
\end{align*}
Thus, there is a constant $d_{n}$ that satisfies
\begin{align*}
\psi'_{+,1}(z_{n})-b_{n}\psi'_{-,2}(z_{n})=d_{n}\psi_{-,2}(z_{n}),
\end{align*}
namely,
\begin{align*}
\psi'_{+,1}(z_{n})=d_{n}\psi_{-,2}(z_{n})+b_{n}\psi'_{-,2}(z_{n}).
\end{align*}
Using \eqref{4.4} and \eqref{5.3a}, the follow result can be obtained
\begin{align}\label{5.3b}
\mu'_{+,1}(z_{n})=e^{-2i\theta_{n}}\left[(d_{n}-2i\theta'_{n}b_{n})\mu_{-,2}(z_{n})
+b_{n}\mu'_{-,2}(z_{n})\right],
\end{align}
here $\theta_{n}=\theta(z_{n})$, $\theta'_{n}=\theta'(z_{n})$.

In the same way, we also get the follow relations
\begin{align}
&\mu_{+,2}(z^{*}_{n})=\tilde{b}_{n}e^{2i\theta^{*}_{n}}\mu_{-,1}(z^{*}_{n}),\label{5.3c}\\
&\mu'_{+,2}(z^{*}_{n})=e^{2i\theta^{*}_{n}}\left[(\tilde{d}_{n}+2i\theta'^{*}_{n}\tilde{b}_{n})
\mu_{-,1}(z^{*}_{n})+\tilde{b}_{n}\mu'_{-,1}(z^{*}_{n})\right],\label{5.3d}\\
&\mu_{+,1}(\check{z}_{n})=\check{b}_{n}e^{-2i\check{\theta}_{n}}\mu_{-,2}(\check{z}_{n}),\label{5.3e}\\
&\mu'_{+,1}(\check{z}_{n})=e^{-2i\check{\theta}_{n}}\left[(\check{d}_{n}-2i\check{\theta}'_{n}
\check{b}_{n})\mu_{-,2}(\check{z}_{n})+\check{b}_{n}\mu'_{-,2}(\check{z}_{n})\right],\label{5.3f}\\
&\mu_{+,2}(\hat{z}_{n})=\hat{b}_{n}e^{2i\hat{\theta}_{n}}\mu_{-,1}(\hat{z}_{n}),\label{5.3g}\\
&\mu'_{+,2}(\hat{z}_{n})=e^{2i\hat{\theta}_{n}}\left[(\hat{d}_{n}+2i\hat{\theta}'_{n}\hat{b}_{n})
\mu_{-,1}(\hat{z}_{n})+\hat{b}_{n}\mu'_{-,1}(\hat{z}_{n})\right],\label{5.3h}
\end{align}
where
\begin{align*}
&\theta^{*}_{n}=\theta(z^{*}_{n}),\quad~~ \check{\theta}_{n}=\theta(\check{z}_{n}), ~\quad \hat{\theta}_{n}=\theta(\hat{z}_{n}),\\
&\theta^{'*}_{n}=\theta'(z^{*}_{n}),\quad \check{\theta}'_{n}=\theta'(\check{z}_{n}),\quad \hat{\theta}'_{n}=\theta'(\hat{z}_{n}),
\end{align*}
and $\tilde{b}_{n}, \tilde{d}_{n}, \check{b}_{n}, \check{d}_{n}, \hat{b}_{n}$ and $\hat{d}_{n}$ are constants.

Similar to the proof of \eqref{4.8}, we also have the following relationship
\begin{gather}\label{S1}
\begin{gathered}
\hat{b}_{n}=-\frac{q^{*}_{-}}{q_{+}}b_{n}, \quad
-\check{b}_{n}=\hat{b}^{*}_{n},\\
d^{*}_{n}=-\tilde{d}_{n},\quad
\hat{d}_{n}=-\frac{q^{*}_{-}z^{2}_{n}}{q_{+}q^{2}_{0}}d_{n},\quad
\check{d}_{n}=-\hat{d}^{*}_{n}.
\end{gathered}
\end{gather}

If $f$ and $g$ is analyzed in $\Omega\in\mathbb{C}$, and $g$ has a double zero at $z_{0}\in\Omega$ and $f(z_{0})\neq0$, then
\begin{align*}
\mathop{Res}_{z=z_{0}}\left[\frac{f}{g}\right]=\frac{2f'(z_{0})}{g''(z_{0})}-
\frac{2f(z_{0})g'''(z_{0})}{3(g''(z_{0}))^{2}},\quad
\mathop{P_{-2}}_{z=z_{0}}\left[\frac{f}{g}\right]=\frac{2f(z_{0})}{g''(z_{0})}.
\end{align*}
Thus
\begin{align}
&\mathop{P_{-2}}_{z=z_{n}}\left[\frac{\mu_{+,1}}{s_{11}}\right]=\frac{2\mu_{+,1}(z_{n})}{s''_{11}(z_{n})}
=\frac{2b_{n}}{s''_{11}(z_{n})}e^{-2i\theta_{n}}\mu_{-,2}(z_{n}),\\
&\mathop{Res}_{z=z_{n}}\left[\frac{\mu_{+,1}}{s_{11}}\right]=
\frac{2b_{n}}{s''_{11}(z_{n})}e^{-2i\theta_{n}}\left[\mu'_{-,2}(z_{n})+
\left(\frac{d_{n}}{b_{n}}-2i\theta'_{n}-\frac{s'''_{11}(z_{n})}{3s''_{11}(z_{n})}\right)\mu_{-,2}(z_{n})\right].
\end{align}

Let $A_{n}=\frac{2b_{n}}{s''_{11}(z_{n})}$ and $B_{n}=\frac{d_{n}}{b_{n}}-\frac{s'''_{11}(z_{n})}{3s''_{11}(z_{n})}$, then the above equations change into
\begin{align}
&\mathop{P_{-2}}_{z=z_{n}}\left[\frac{\mu_{+,1}}{s_{11}}\right]=
A_{n}e^{-2i\theta_{n}}\mu_{-,2}(z_{n}),\\
&\mathop{Res}_{z=z_{n}}\left[\frac{\mu_{+,1}}{s_{11}}\right]=
A_{n}e^{-2i\theta_{n}}\left[\mu'_{-,2}(z_{n})+
\left(B_{n}-2i\theta'_{n}\right)\mu_{-,2}(z_{n})\right].
\end{align}
In the same way, we also get the similar relations as follows
\begin{align}
&\mathop{P_{-2}}_{z=z^{*}_{n}}\left[\frac{\mu_{+,2}}{s_{22}}\right]=
\tilde{A}_{n}e^{2i\theta^{*}_{n}}\mu_{-,1}(z^{*}_{n}),\\
&\mathop{Res}_{z=z^{*}_{n}}\left[\frac{\mu_{+,2}}{s_{22}}\right]=
\tilde{A}_{n}e^{2i\theta^{*}_{n}}\left[\mu'_{-,1}(z^{*}_{n})+
\left(\tilde{B}_{n}+2i\theta'^{*}_{n}\right)\mu_{-,1}(z^{*}_{n})\right],\\
&\mathop{P_{-2}}_{z=\check{z}_{n}}\left[\frac{\mu_{+,1}}{s_{11}}\right]=
\check{A}_{n}e^{-2i\check{\theta}_{n}}\mu_{-,2}(\check{z}_{n}),\\
&\mathop{Res}_{z=\check{z}_{n}}\left[\frac{\mu_{+,1}}{s_{11}}\right]=
\check{A}_{n}e^{-2i\check{\theta}_{n}}\left[\mu'_{-,2}(\check{z}_{n})+
\left(\check{B}_{n}-2i\check{\theta}'_{n}\right)\mu_{-,2}(\check{z}_{n})\right],\\
&\mathop{P_{-2}}_{z=\hat{z}_{n}}\left[\frac{\mu_{+,2}}{s_{22}}\right]=
\hat{A}_{n}e^{2i\hat{\theta}_{n}}\mu_{-,1}(\hat{z}_{n}),\\
&\mathop{Res}_{z=\hat{z}_{n}}\left[\frac{\mu_{+,2}}{s_{22}}\right]=
\hat{A}_{n}e^{2i\hat{\theta}_{n}}\left[\mu'_{-,1}(\hat{z}_{n})+
\left(\hat{B}_{n}+2i\hat{\theta}'_{n}\right)\mu_{-,1}(\hat{z}_{n})\right],
\end{align}
where
\begin{align*}
\tilde{A}_{n}=\frac{2\tilde{b}_{n}}{s''_{22}(z^{*}_{n})}, \quad
\tilde{B}_{n}=\frac{\tilde{d}_{n}}{\tilde{b}_{n}}
-\frac{s'''_{22}(z^{*}_{n})}{3s''_{22}(z^{*}_{n})},\\
\check{A}_{n}=\frac{2\check{b}_{n}}{s''_{11}(\check{z}_{n})}, \quad
\check{B}_{n}=\frac{\check{d}_{n}}{\check{b}_{n}}
-\frac{s'''_{11}(\check{z}_{n})}{3s''_{11}(\check{z}_{n})},\\
\hat{A}_{n}=\frac{2\hat{b}_{n}}{s''_{22}(\hat{z}_{n})}, \quad
\hat{B}_{n}=\frac{\hat{d}_{n}}{\hat{b}_{n}}
-\frac{s'''_{22}(\hat{z}_{n})}{3s''_{22}(\hat{z}_{n})}.
\end{align*}
According to the symmetry in \eqref{4.8} and \eqref{S1}, the following relation also can be obtained
\begin{align*}
&\tilde{A}_{n}=-A^{*}_{n}, \quad\quad\quad~\tilde{B}_{n}=B^{*}_{n},\\
&\hat{A}_{n}=-\frac{q^{4}_{0}q^{*}_{-}}{z^{4}_{n}q_{-}}A_{n},\quad
\hat{B}_{n}=\frac{q^{2}_{0}}{\hat{z}^{2}_{n}}B_{n}+\frac{2}{\hat{z}_{n}},\\
&\check{A}_{n}=-\hat{A}^{*}_{n}, \quad\quad\quad~\check{B}_{n}=\hat{B}^{*}_{n}.
\end{align*}

\subsection{Reconstruction the formula for potential}

For convenience, let
\begin{align*}
\xi_{n}:=z_{n}, \quad \xi_{N+n}:=\check{z}_{n}, \quad \xi^{*}_{n}:=\hat{z}_{n}, \quad \xi^{*}_{N+n}:=z^{*}_{n}.
\end{align*}
Similar to the case of the simple poles, to get a regular RHP, we need to subtract the asymptotic properties and the pole contributions. The jump condition becomes
\begin{align}\label{5.5}
\begin{aligned}
M^{-}&(z)-\mathbb{I}-\frac{i}{z}\sigma_{3}Q_{-}-\sum^{2N}_{n=1}
\left\{\frac{\mathop{Res}\limits_{z=\xi^{*}_{n}}M^{-}(z)}{z-\xi^{*}_{n}}
+\frac{\mathop{P_{-2}}\limits_{z=\xi^{*}_{n}}M^{-}(z)}{(z-\xi^{*}_{n})^{2}}
\frac{\mathop{Res}\limits_{z=\xi_{n}}M^{+}(z)}{z-\xi_{n}}
+\frac{\mathop{P_{-2}}\limits_{z=\xi_{n}}M^{+}(z)}{(z-\xi_{n})^{2}}\right\}\\
=&M^{+}(z)-\mathbb{I}-\frac{i}{z}\sigma_{3}Q_{-}-\sum^{2N}_{n=1}
\left\{\frac{\mathop{Res}\limits_{z=\xi^{*}_{n}}M^{-}(z)}{z-\xi^{*}_{n}}
+\frac{\mathop{P_{-2}}\limits_{z=\xi^{*}_{n}}M^{-}(z)}{(z-\xi^{*}_{n})^{2}}
\frac{\mathop{Res}\limits_{z=\xi_{n}}M^{+}(z)}{z-\xi_{n}}
+\frac{\mathop{P_{-2}}\limits_{z=\xi_{n}}M^{+}(z)}{(z-\xi_{n})^{2}}\right\}\\
&-M^{+}(z)G(z).
\end{aligned}
\end{align}

During the defining of $M^{\pm}(z)$, we have
\begin{align}\label{5.6}
\begin{gathered}
\mathop{Res}\limits_{z=\xi^{*}_{n}}M^{-}(z)=\left(0, \mathop{Res}\limits_{z=\xi^{*}_{n}}\left[\frac{\mu_{+,2}}{s_{22}}\right]\right),\quad
\mathop{P_{-2}}\limits_{z=\xi^{*}_{n}}M^{-}(z)=\left(0,
\mathop{P_{-2}}\limits_{z=\xi^{*}_{n}}\left[\frac{\mu_{+,2}}{s_{22}}\right]\right),\\
\mathop{Res}\limits_{z=\xi_{n}}M^{+}(z)=\left( \mathop{Res}\limits_{z=\xi_{n}}\left[\frac{\mu_{+,1}}{s_{11}}\right], 0\right),\quad
\mathop{P_{-2}}\limits_{z=\xi_{n}}M^{+}(z)=\left( \mathop{P_{-2}}\limits_{z=\xi_{n}}\left[\frac{\mu_{+,1}}{s_{11}}\right], 0\right).
\end{gathered}
\end{align}
It is easy to see that the left side of \eqref{5.6} is analytic in $D^{-}$, and the right side except the last item is analytic in $D^{+}$. Resorting to Plemelj' formulae, the solution of RHP is
\begin{align}\label{5.7}
\begin{aligned}
M(z)
=&\mathbb{I}+\frac{i}{z}\sigma_{3}Q_{-}+\sum^{2N}_{n=1}
\left\{\frac{\mathop{Res}\limits_{z=\xi^{*}_{n}}M^{-}(z)}{z-\xi^{*}_{n}}
+\frac{\mathop{P_{-2}}\limits_{z=\xi^{*}_{n}}M^{-}(z)}{(z-\xi^{*}_{n})^{2}}
\frac{\mathop{Res}\limits_{z=\xi_{n}}M^{+}(z)}{z-\xi_{n}}
+\frac{\mathop{P_{-2}}\limits_{z=\xi_{n}}M^{+}(z)}{(z-\xi_{n})^{2}}\right\}\\
&+\frac{1}{2\pi i}\int_{\Sigma}\frac{M^{+}(s)G(s)}{s-z}ds,\quad z\in\mathbb{C}\backslash\Sigma.
\end{aligned}
\end{align}

Evaluating the first column of \eqref{5.7} at $z=\xi^{*}_{k}$, we get
\begin{align}\label{5.8}
\begin{aligned}
\mu_{-,1}(\xi^{*}_{k})=&\left(\begin{array}{c}
                         1 \\
                         \frac{q_{-}}{\xi^{*}_{k}}
                       \end{array}\right)
+\frac{1}{2\pi i}\int_{\Sigma}\frac{((M^{+}G)(s))_{1}}{s-\xi^{*}_{k}}ds\\
&+\sum^{2N}_{n=1}\left[C_{n}(\xi^{*}_{k})\mu'_{-,2}(\xi_{n})+C_{n}(\xi^{*}_{k})
\left(D_{n}+\frac{1}{\xi^{*}_{k}-\xi_{n}}\right)\mu_{-,2}(\xi_{n})\right],
\end{aligned}
\end{align}
here $C_{n}(z)=\frac{A_{n}e^{-2i\theta_{n}}}{z-\xi_{n}}$ and
$D_{n}=B_{n}-2i\theta'_{n}$.

By the symmetry of \eqref{Q45}, we know that $\mu_{-,1}(\xi^{*}_{k})=-\frac{\xi_{k}}{q^{*}_{-}}\mu_{-,2}(\xi_{k})$, then \eqref{5.8} can be rewritten as
\begin{align}\label{5.9}
\begin{aligned}
0=&\left(\begin{array}{c}
                         1 \\
                         \frac{q_{-}}{\xi^{*}_{k}}
                       \end{array}\right)
+\frac{1}{2\pi i}\int_{\Sigma}\frac{((M^{+}G)(s))_{1}}{s-\xi^{*}_{k}}ds\\
&+\sum^{2N}_{n=1}\left\{C_{n}(\xi^{*}_{k})\mu'_{-,2}(\xi_{n})+C_{n}(\xi^{*}_{k})
\left[\left(D_{n}+\frac{1}{\xi^{*}_{k}-\xi_{n}}\right)+\frac{\delta_{kn}\xi_{k}}{q^{*}_{-}}\right]
\mu_{-,2}(\xi_{n})\right\},
\end{aligned}
\end{align}
where $\delta_{kn}$ is Kronecker delta. In order to obtain further results at the derivative of the eigenfunction, we derive the derivative of \eqref{5.9} with respect to $z$ at $z=\xi^{*}_{k}$ and obtain
\begin{align}\label{5.10}
\begin{aligned}
0=&\left(\begin{array}{c}
                         1 \\
                         -\frac{q_{-}}{(\xi^{*}_{k})^{2}}
                       \end{array}\right)
+\frac{1}{2\pi i}\int_{\Sigma}\frac{((M^{+}G)(s))_{1}}{(s-\xi^{*}_{k})^{2}}ds\\
&-\sum^{2N}_{n=1}\left\{\left(\frac{C_{n}(\xi^{*}_{k})}{\xi^{*}_{k}-\xi_{n}}
-\frac{\xi^{*}_{k}\delta_{kn}}{q^{2}_{0}q^{*}_{-}}\right)\mu'_{-,2}(\xi_{n})
+\left[\frac{C_{n}(\xi^{*}_{k})}{\xi^{*}_{k}-\xi_{n}}
\left(D_{n}+\frac{2}{\xi^{*}_{k}-\xi_{n}}\right)-\frac{\xi^{2}_{k}}{q^{2}_{0}q^{*}_{-}}\right]
\mu_{-,2}(\xi_{n})\right\}.
\end{aligned}
\end{align}
Let $M=M^{-}$, using the same method as in the simple poles get
\begin{align}\label{5.11}
q=q_{-}+\frac{i}{2\pi}\int_{\Sigma}((M^{+}G)(s))_{21}ds-
\sum^{2N}_{n=1}A_{n}e^{-2i\theta_{n}}\left[
\mu'_{-,22}(\xi_{n})+D_{n}\mu_{-,22}(\xi_{n})\right].
\end{align}

\subsection{Trace formulate and theta condition}

According to analyticity of $s_{11}$ and $s_{22}$ and the discrete spectrum, it is obvious to know that $\beta^{+}_{2}(z)$ has the double poles in $z_{n}$ and $-q^{2}_{0}/z^{*}_{n}$, and $\beta^{-}_{2}(z)$ in $z^{*}_{n}$ and $-q^{2}_{0}/z_{n}$ respectively for $n=1,2,\cdots,N$. For
\begin{align}\label{5.12}
\begin{split}
\beta^{+}_{2}(z)=s_{11}(z)\prod^{N}_{n=1}\frac{(z-z^{*}_{n})^{2}(z+q^{2}_{0}/z_{n})^{2}}
{(z-z_{n})^{2}(z+q^{2}_{0}/z^{*}_{n})^{2}},\\
\beta^{-}_{2}(z)=s_{22}(z)\prod^{N}_{n=1}\frac{(z-z_{n})^{2}(z+q^{2}_{0}/z^{*}_{n})^{2}}
{(z-z^{*}_{n})^{2}(z+q^{2}_{0}/z_{n})^{2}},
\end{split}
\end{align}
$\beta^{+}_{2}(z)$ and $\beta^{-}_{2}(z)$ are analytic in $D^{+}$ and $D^{-}$ respectively
 and have no zero point. Furthermore, their asymptotic properties are the same as $s_{11}(z)$ and $s_{22}(z)$ when $z\rightarrow\infty$. Thus there have $\beta^{+}_{2}(z)\beta^{-}_{2}(z)=s_{11}(z)s_{22}(z)$ for $z\in\Sigma$. Again by $detS(z)=1$, we know
\begin{align}
\beta^{+}_{2}(z)\beta^{-}_{2}(z)=\frac{1}{1+\rho(z)\rho^{*}(z^{*})},\quad z\in\Sigma.
\end{align}
Let's take the logarithms for the above equation and apply the Plemelj's formulae, it is easy to obtain that
\begin{align}\label{5.13}
\log\beta^{\pm}_{2}(z)=-\frac{1}{2\pi i}\int_{\Sigma}\frac{\log[1+\rho(\zeta)\rho^{*}(\zeta^{*})]}
{\zeta-(z\pm i0)}d\zeta, \quad z\in D^{\pm}.
\end{align}

Similar to the calculation in the simple pole, trace formulate and theta condition are
\begin{align}\label{5.14}
\begin{split}
&s_{11}(z)=exp\left[-\frac{1}{2\pi i}\int_{\Sigma}\frac{\log[1+\rho(\zeta)\rho^{*}(\zeta^{*})]}
{\zeta-z}d\zeta\right]\prod^{N}_{n=1}\frac{(z-z_{n})^{2}(z+q^{2}_{0}/z^{*}_{n})^{2}}
{(z-z^{*}_{n})^{2}(z+q^{2}_{0}/z_{n})^{2}},\quad z\in D^{+},\\
&s_{22}(z)=exp\left[-\frac{1}{2\pi i}\int_{\Sigma}\frac{\log[1+\rho(\zeta)\rho^{*}(\zeta^{*})]}
{\zeta-z}d\zeta\right]\prod^{N}_{n=1}\frac{(z-z^{*}_{n})^{2}(z+q^{2}_{0}/z_{n})^{2}}
{(z-z_{n})^{2}(z+q^{2}_{0}/z^{*}_{n})^{2}},\quad z\in D^{-},
\end{split}
\end{align}
and
\begin{align}\label{5.15}
arg\left(\frac{q_{+}}{q_{-}}\right)=\frac{1}{2\pi }\int_{\Sigma}
\frac{\log[1+\rho(\zeta)\rho^{*}(\zeta^{*})]}{\zeta-z}d\zeta
+8\sum^{N}_{n=1}argz_{n}.
\end{align}

\subsection{Reflection-less potential}

In this section, we will consider the case of gradual disappearance of the reflection coefficient, i.e., $G(z)=0$. From \eqref{5.11}, the solution of the equation is
\begin{align}\label{5.16}
q=q_{-}-
\sum^{2N}_{n=1}A_{n}e^{-2i\theta_{n}}\left[
\mu'_{-,22}(\xi_{n})+D_{n}\mu_{-,22}(\xi_{n})\right].
\end{align}

In the following, we will focus on solving $\mu'_{-,22}(\xi_{n})$ and $\mu_{-,22}(\xi_{n})$.

When $G(z)=0$, through \eqref{5.9} and \eqref{5.10} we have
\begin{align}\label{5.17}
\begin{aligned}
0=\frac{q_{-}}{\xi^{*}_{k}}
+\sum^{2N}_{n=1}\left\{C_{n}(\xi^{*}_{k})\mu'_{-,2}(\xi_{n})+C_{n}(\xi^{*}_{k})
\left[\left(D_{n}+\frac{1}{\xi^{*}_{k}-\xi_{n}}\right)+\frac{\delta_{kn}\xi_{k}}{q^{*}_{-}}\right]
\mu_{-,2}(\xi_{n})\right\},
\end{aligned}
\end{align}
and
\begin{align}\label{5.18}
\begin{aligned}
0=\frac{q_{-}}{(\xi^{*}_{k})^{2}}
+\sum^{2N}_{n=1}&\left\{\left(\frac{C_{n}(\xi^{*}_{k})}{\xi^{*}_{k}-\xi_{n}}
-\frac{\xi^{*}_{k}\delta_{kn}}{q^{2}_{0}q^{*}_{-}}\right)\mu'_{-,2}(\xi_{n})\right.\\
&\left.+\left[\frac{C_{n}(\xi^{*}_{k})}{\xi^{*}_{k}-\xi_{n}}
\left(D_{n}+\frac{2}{\xi^{*}_{k}-\xi_{n}}\right)-\frac{\xi^{2}_{k}}{q^{2}_{0}q^{*}_{-}}\right]
\mu_{-,2}(\xi_{n})\right\}.
\end{aligned}
\end{align}

Defining
\begin{align*}
&X_{n}=\mu_{-,22}(\xi_{n}),\qquad\quad\quad\quad\quad\quad\quad\quad\quad\quad\quad X_{2N+n}=\mu'_{-,22}(\xi_{n}),\\
&V_{n}=\frac{q_{-}}{(\xi_{n})},\qquad\quad\quad\quad\quad\quad\quad\quad\quad\quad\quad\quad~~~ V_{2N+n}=\frac{q_{-}}{(\xi_{n})^{2}},\\
&M_{k,n}=C_{n}(\xi^{*}_{k})\left(D_{n}+\frac{1}{\xi^{*}_{k}-\xi_{n}}\right)
+\frac{\delta_{kn}\xi_{k}}{q^{*}_{-}},\qquad~~~
M_{k,2N+n}=C_{n}(\xi^{*}_{k}),\\
&M_{2N+k,n}=\frac{C_{n}(\xi^{*}_{k})}{\xi^{*}_{k}-\xi_{n}}
\left(D_{n}+\frac{2}{\xi^{*}_{k}-\xi_{n}}\right)-\frac{\xi^{2}_{k}}{q^{2}_{0}q^{*}_{-}},\quad
M_{2N+k,2N+n}=\frac{C_{n}(\xi^{*}_{k})}{\xi^{*}_{k}-\xi_{n}}
-\frac{\xi^{*}_{k}\delta_{kn}}{q^{2}_{0}q^{*}_{-}},
\end{align*}
with $n=1,2,\cdots,2N$.

Therefore, \eqref{5.17} and \eqref{5.18} are expressed in component form as
\begin{align}\label{5.19}
MX=V.
\end{align}
By Cramer's Rule, the solution of \eqref{5.19} is
\begin{align}
X_{n}=\frac{detM^{ext}_{n}}{detM},
\end{align}
where $M^{ext}_{n}=(M_{1},\cdots,M_{n-1},V,M_{n+1},\cdots,M_{2N+n})$.

\subsection{Soliton solutions with the double poles}

In this subsection, we mainly discuss the exact solution images obtained in the previous section and analyze their dynamic behavior. Let $N=1$, the image of the solution of the equation is as follows.
\\

\noindent
{\rotatebox{0}{\includegraphics[width=3.5cm,height=3.5cm,angle=0]{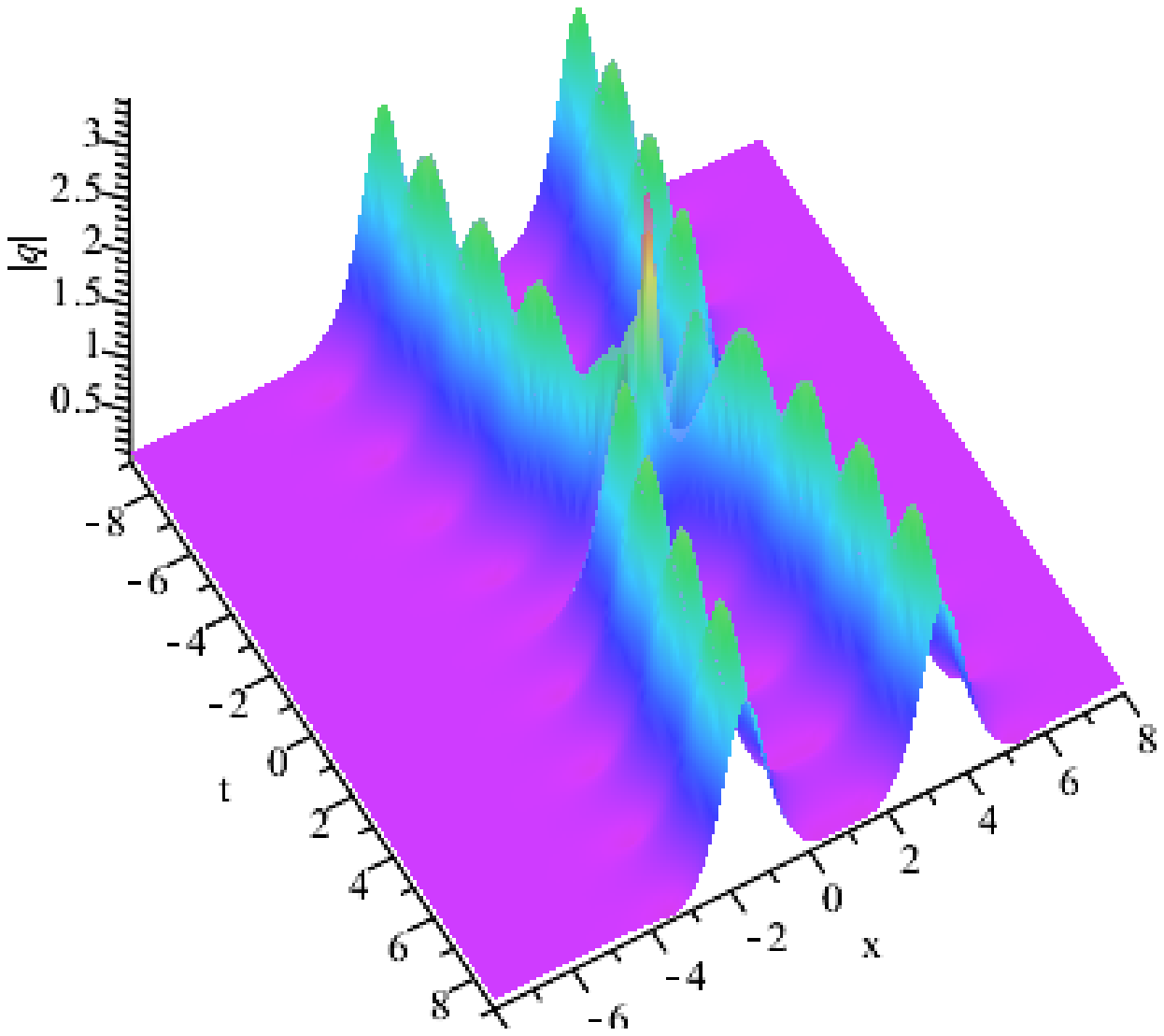}}}
~~~~
{\rotatebox{0}{\includegraphics[width=3.5cm,height=3.5cm,angle=0]{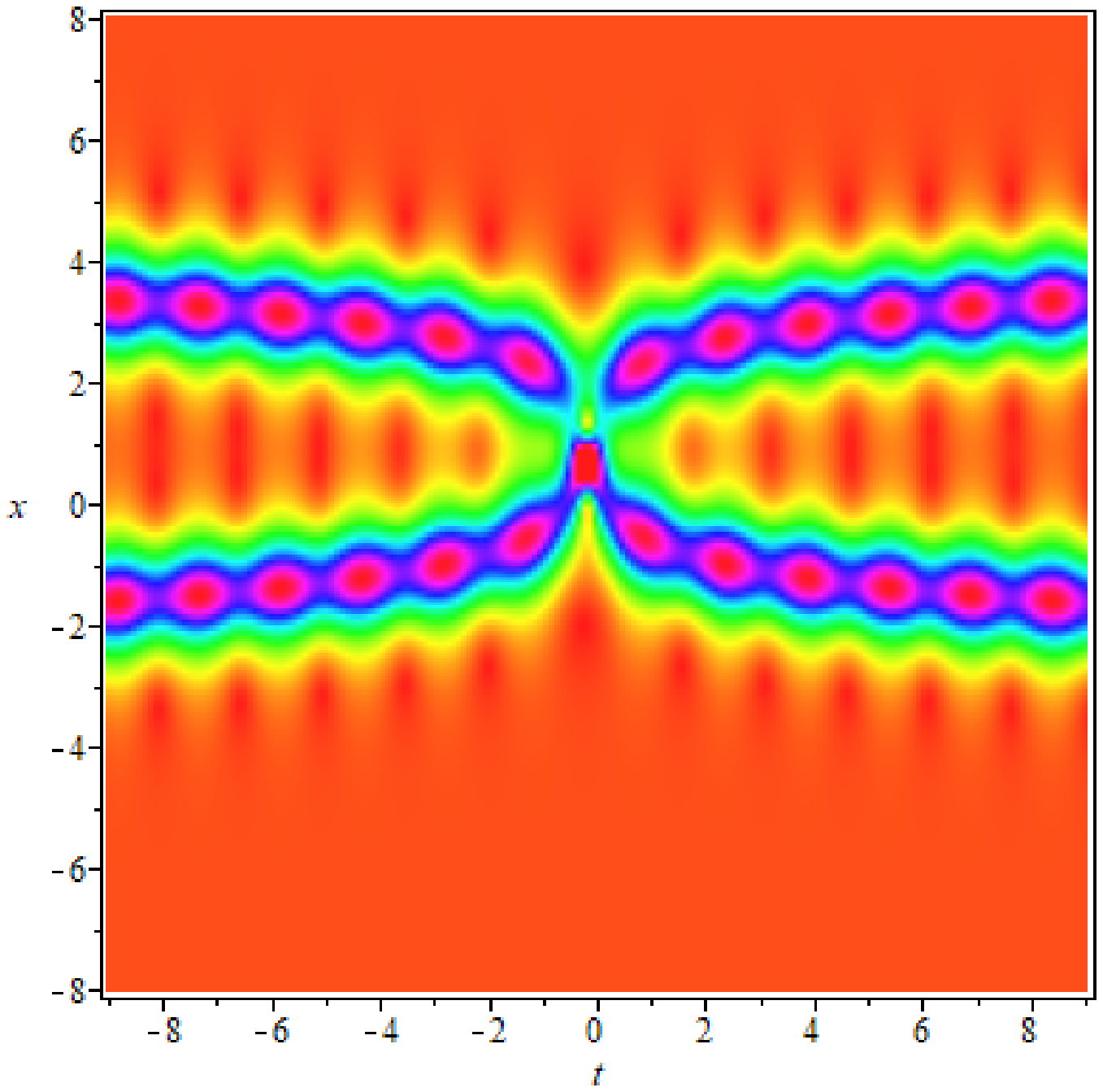}}}
\qquad\quad
{\rotatebox{0}{\includegraphics[width=3.5cm,height=3.5cm,angle=0]{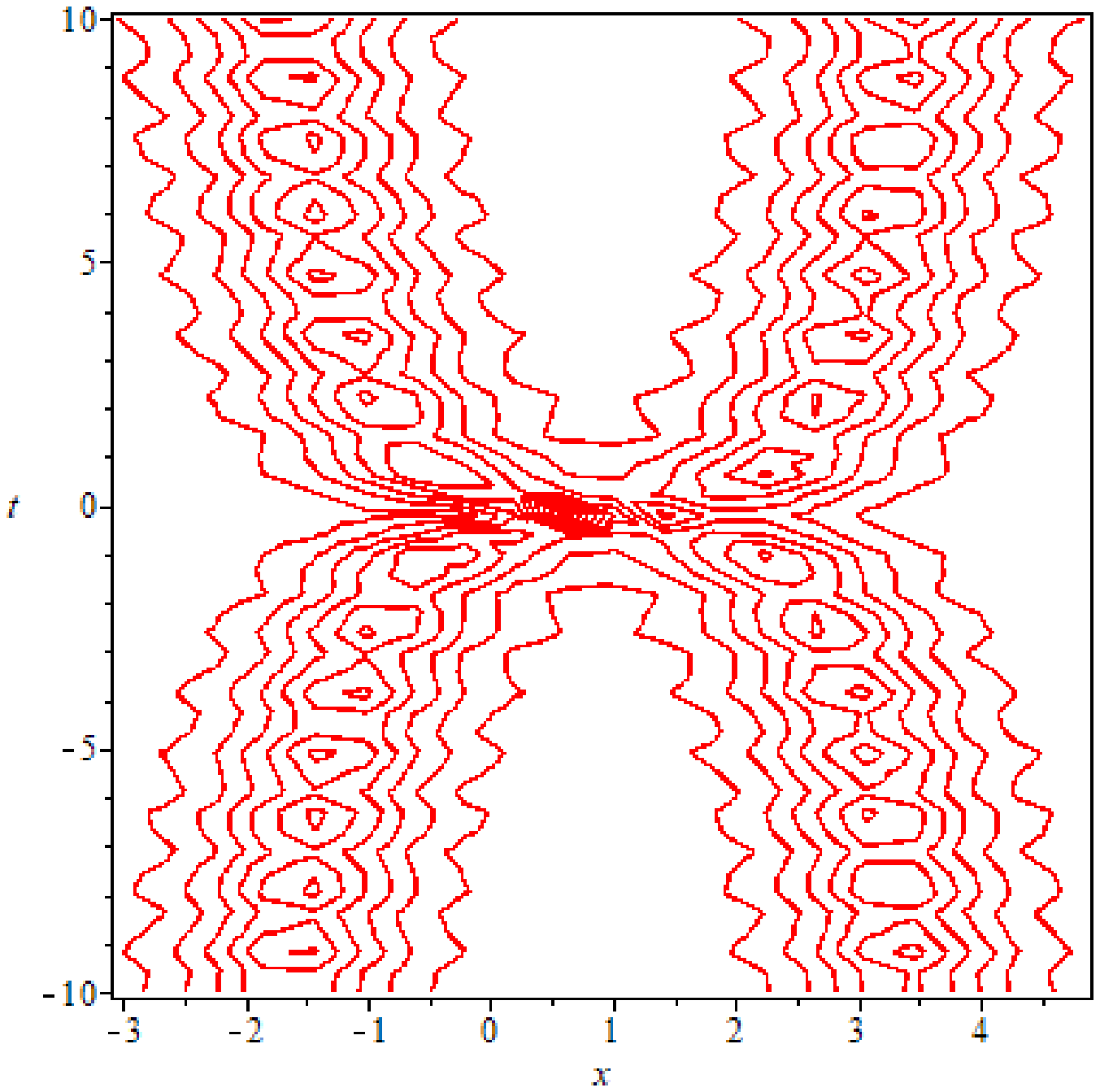}}}

\qquad\quad $(a)$
\qquad\qquad\qquad\qquad\qquad\qquad $(b)$ \quad\qquad\qquad\qquad\qquad\qquad\qquad $(c)$

\noindent
{\rotatebox{0}{\includegraphics[width=3.5cm,height=3.5cm,angle=0]{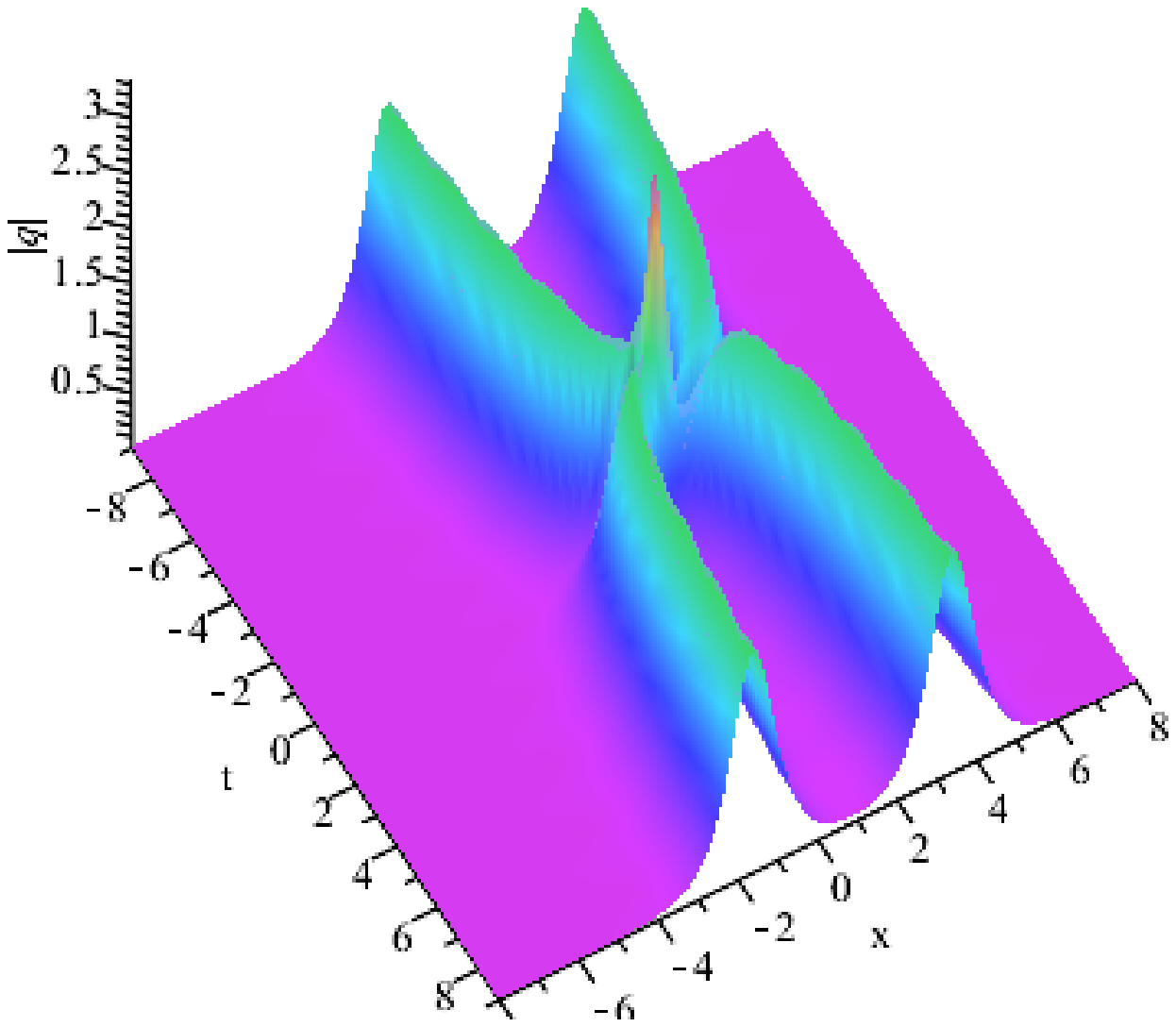}}}
~~~~
{\rotatebox{0}{\includegraphics[width=3.5cm,height=3.5cm,angle=0]{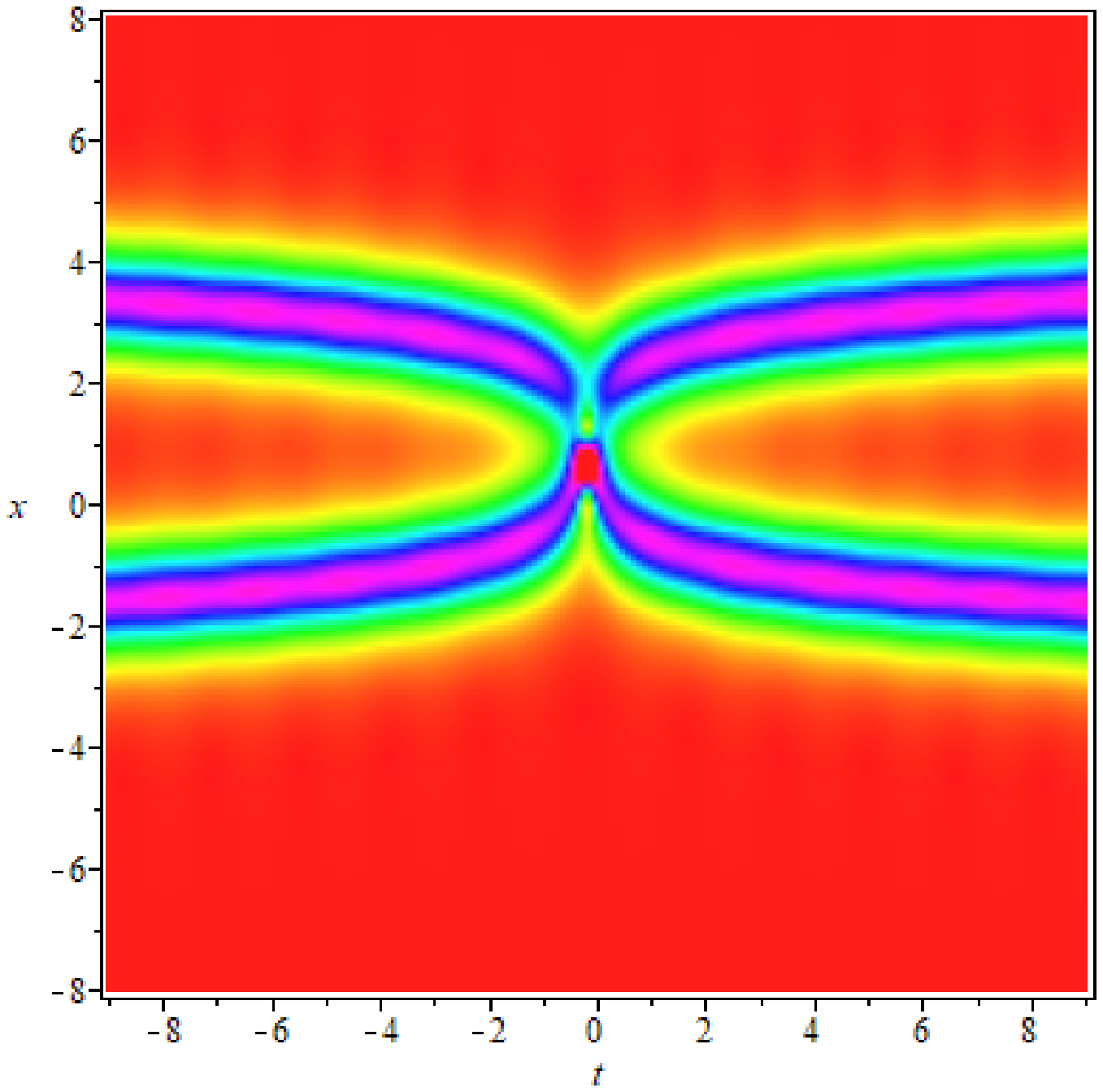}}}
\qquad\quad
{\rotatebox{0}{\includegraphics[width=3.5cm,height=3.5cm,angle=0]{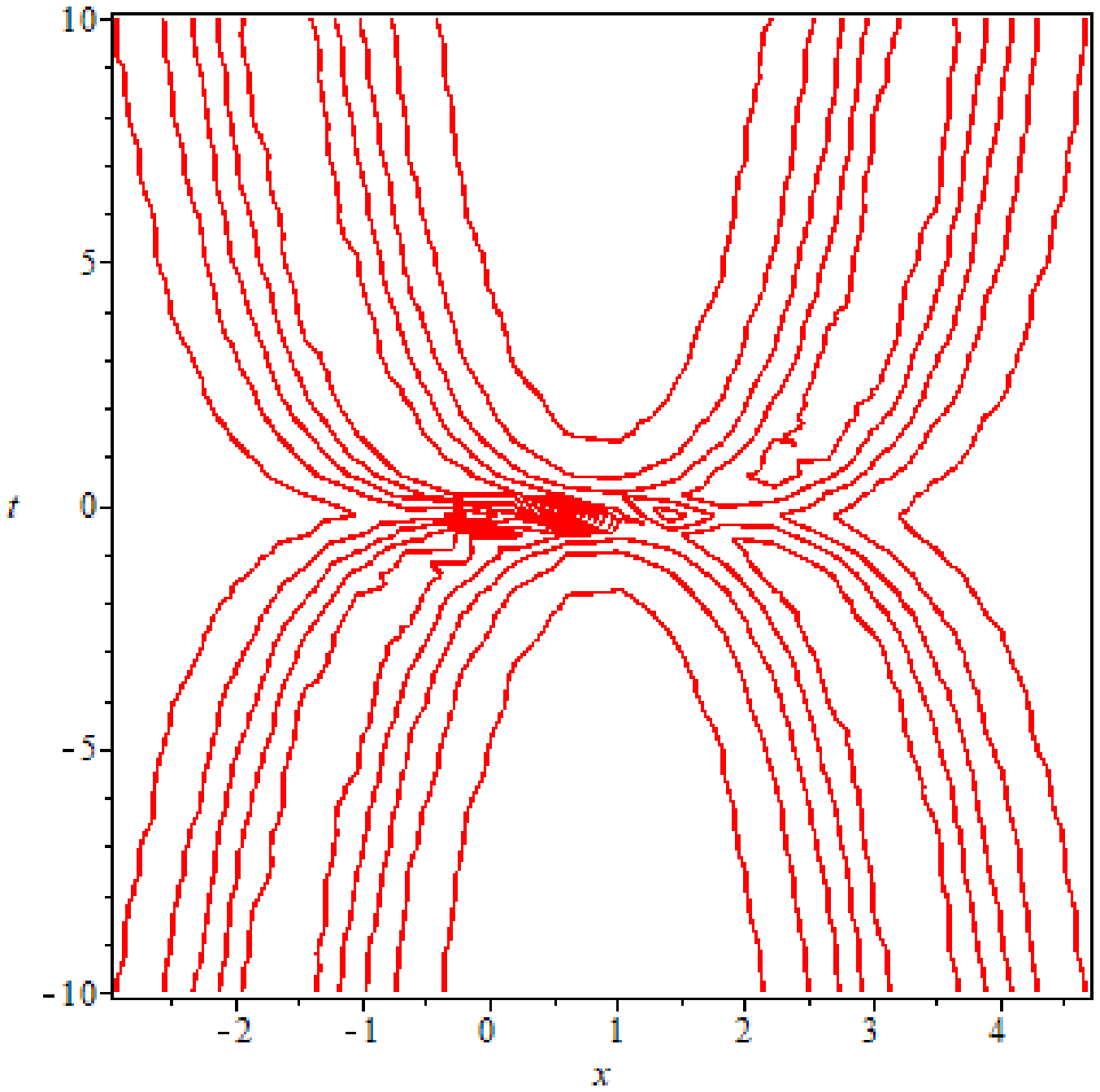}}}

\qquad\quad $(d)$
\qquad\qquad\qquad\qquad\qquad\qquad $(e)$ \quad\qquad\qquad\qquad\qquad\qquad\qquad $(f)$

\noindent
{\rotatebox{0}{\includegraphics[width=3.5cm,height=3.5cm,angle=0]{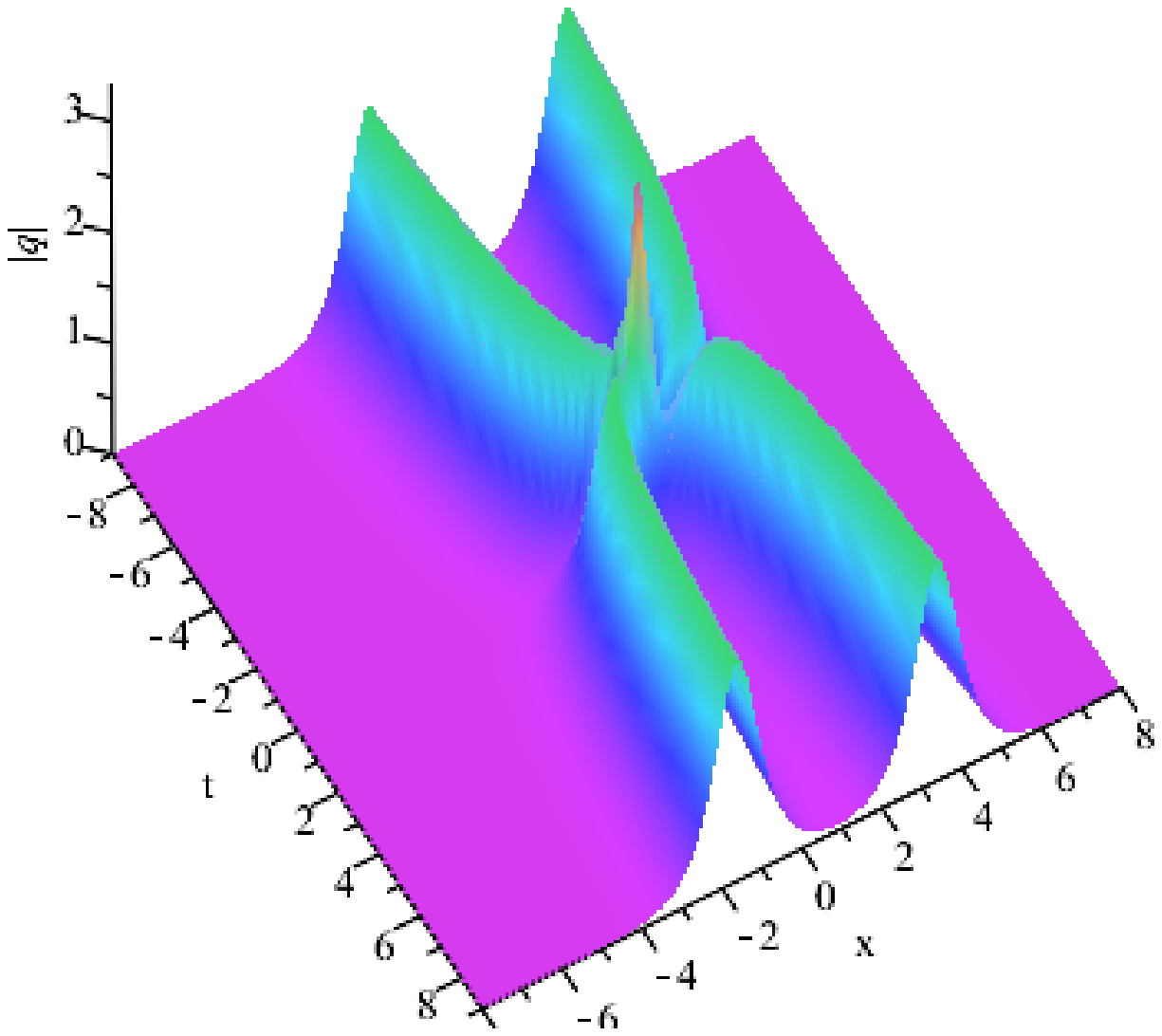}}}
~~~~
{\rotatebox{0}{\includegraphics[width=3.5cm,height=3.5cm,angle=0]{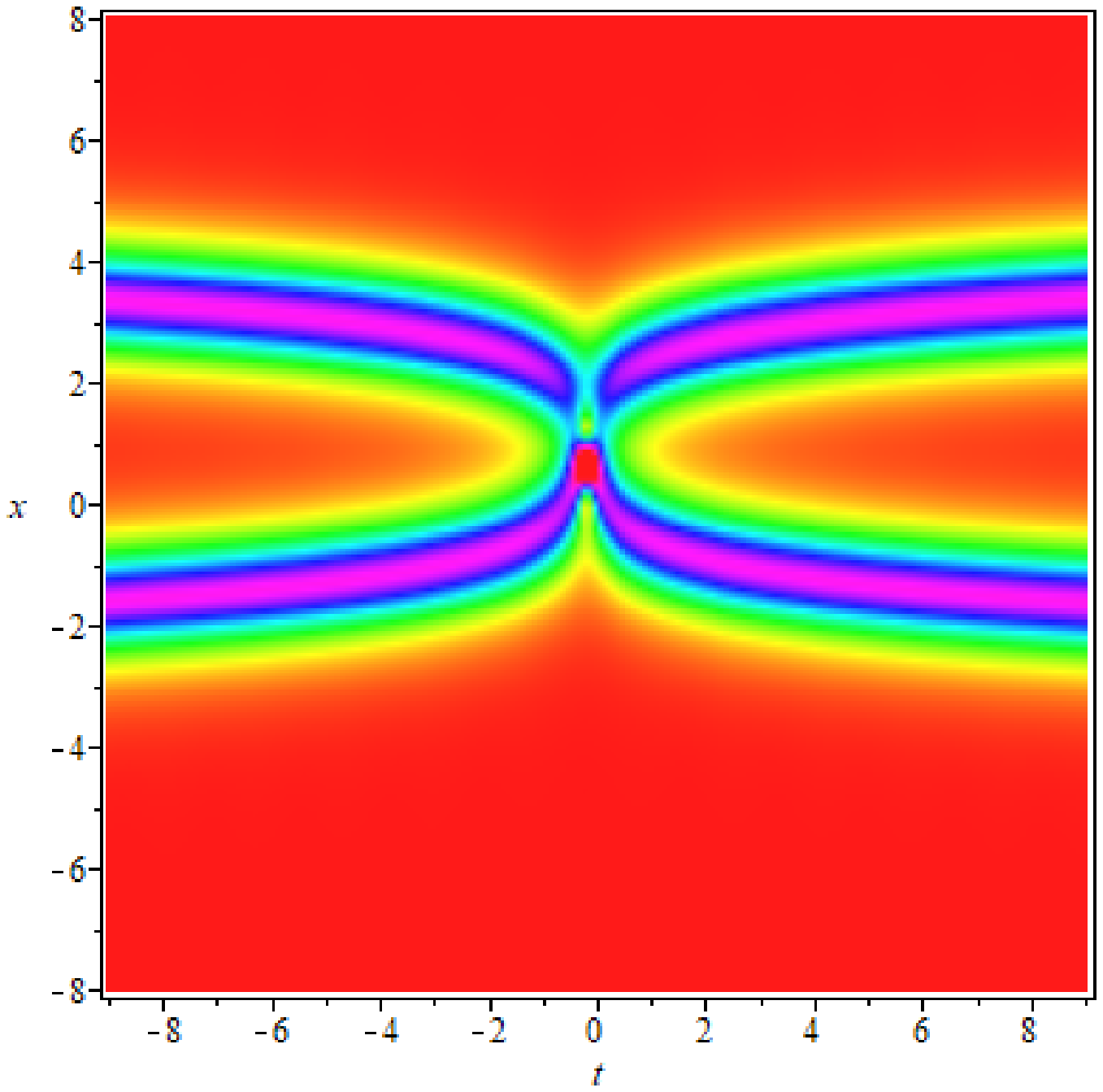}}}
\qquad\quad
{\rotatebox{0}{\includegraphics[width=3.5cm,height=3.5cm,angle=0]{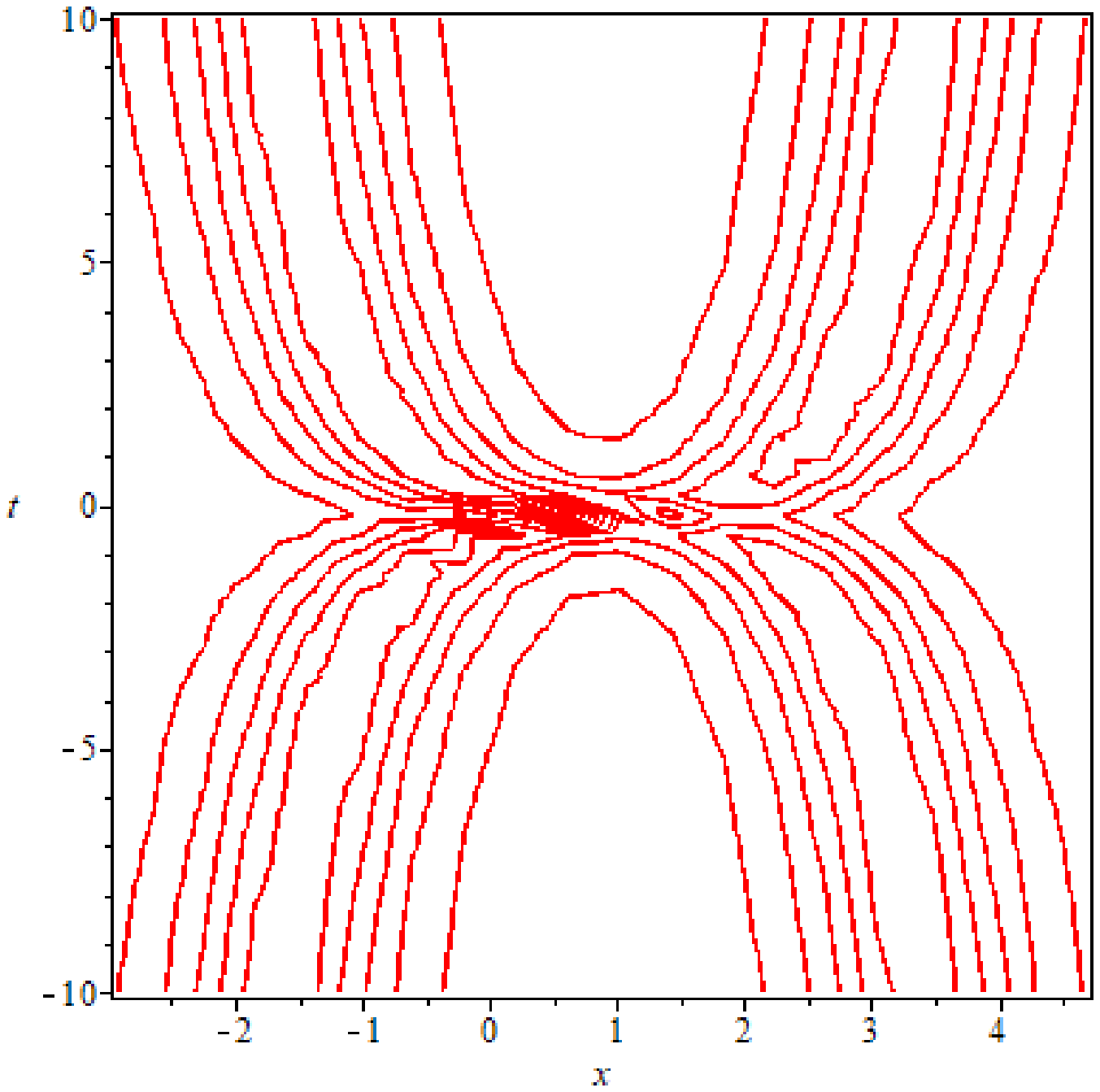}}}

\qquad\quad $(g)$
\qquad\qquad\qquad\qquad\qquad\qquad $(h)$ \quad\qquad\qquad\qquad\qquad\qquad\qquad $(i)$\\
\noindent { \small \textbf{Figure 5.}
The breather wave solution for $|q|$ with the parameters selection
$\delta=1, \xi_{1}=2i$.
$\textbf{(a)}$ the soliton solution with $q_{-}=0.1$,
$\textbf{(b)}$ density plot corresponding to (a),
$\textbf{(c)}$ the contour line of the soliton solution corresponding to (a),
$\textbf{(d)}$ the soliton solution with $q_{-}=0.01$,
$\textbf{(e)}$ density plot corresponding to (d),
$\textbf{(f)}$ the contour line of the soliton solution corresponding to (d),
$\textbf{(g)}$ the soliton solution with $q_{-}=0.001$,
$\textbf{(h)}$ density plot corresponding to (g),
$\textbf{(i)}$ the contour line of the soliton solution corresponding to (g).}\\

In Fig.5, we separately discuss the images of $q_{-}$ at different values, and find that when $q_{-}$ approaches 0, the breathing wave approaches a bright soliton.

\section{Conclusions and discussions}

In this work, we mainly solve the exact solution of the sixtic Schr\"{o}dinger equation based on the generalized Riemann-Hilbert problem with non-zero boundary values.
Since the non-zero boundary considerations is considered here, the sixtic Schr\"{o}dinger equation needs to be transformed first, so that the non-zero boundary does not depend on $t$. In order to avoid the multi-value problem, Riemann surface and uniformization coordinate  are introduced to transform the original spectrum problem into a new spectrum problem. On this basis, we discuss the analysis, symmetry and asymptotic properties of the Jost functions and the scattering matrix, which will prepare for the establishment of a generalized Riemann-Hilbert problem.
In inverse scattering, it is divided into two cases: simple pole and double pole. Doing the same process, we reconstruct the formula for potential according to the discrete spectrum and residue conditions, and discuss the trace formula and theta condition. Finally, the soliton solutions of the equation are solved without reflection potential.
In addition, we have made vivid image descriptions of the exact solutions in simple poles and double poles respectively by selecting appropriate parameters, which is of great help to the study of the propagation of the equation and its dynamic behavior. We also hope that our results will be of great significance to the study of Schr\"{o}dinger equation.

\section*{Acknowledgements}

This work was supported by  the National Natural Science Foundation of China under Grant No. 11975306, the Natural Science Foundation of Jiangsu Province under Grant No. BK20181351, the Six Talent Peaks Project in Jiangsu Province under Grant No. JY-059, the Fundamental Research Fund for the Central Universities under the Grant Nos. 2019ZDPY07 and 2019QNA35, and the Assistance Program for Future Outstanding Talents of China University of Mining and Technology under Grant No.2020WLJCRCZL074.

\end{document}